\documentclass[smallcondensed, final]{svjour3}[]
 
\usepackage{natbib}
\setcitestyle{authoryear,open={(},close={)}}

\usepackage{xcolor}
\usepackage{amsmath}
\usepackage{latexsym}
\usepackage{amssymb}
\usepackage[english]{babel}
\usepackage{bbm}
\usepackage{graphicx}
\usepackage{float}
\usepackage[title]{appendix}
\usepackage[multiple]{footmisc}
 \usepackage[utf8]{inputenc}
\usepackage{bm}
\usepackage{mathtools}
\usepackage[multiple]{footmisc}
\usepackage{color}
\usepackage{soul}
\usepackage{dsfont}
\usepackage{tikz}
\usepackage{mathptmx}
\usetikzlibrary{calc}
\usetikzlibrary{graphs}
\usetikzlibrary{arrows}
\usetikzlibrary{shapes}
\usetikzlibrary{decorations.pathmorphing}
\usetikzlibrary{trees}
\usetikzlibrary{snakes}
\usetikzlibrary{topaths}
\usetikzlibrary{plotmarks}
\usetikzlibrary{backgrounds}

\usepackage{float}
\usepackage{geometry}
\usepackage{multicol}
\usepackage{natbib}

\newtheorem{assumption}{Assumption}

\begin{document}

\title{ Bias-Optimal Vol-of-Vol Estimation: The Role of Window Overlapping   }

 \author{Giacomo Toscano \and Maria Cristina Recchioni}
 \institute{ Giacomo Toscano (corresponding author) \at University of Firenze and Scuola Normale Superiore, Italy - email: giacomo.toscano@unifi.it; address: Via delle Pandette 9, 50127 Firenze, Italy    \and Maria Cristina Recchioni \at Università Politecnica delle Marche, Italy - email: m.c.recchioni@univpm.it; address: Piazzale Martelli 8, 60121 Ancona, Italy}

 \maketitle

\begin{abstract}

The simplest and most natural vol-of-vol estimator, the pre-estimated spot variance based realized variance (PSRV), is typically plagued by a large finite-sample bias. In this paper, we analytically show that allowing for the overlap of consecutive local windows to pre-estimate the spot variance may correct for this bias. In particular, we provide a feasible rule for the bias-optimal selection of the length of local windows when the volatility is a CKLS process. The effectiveness of this rule for practical applications is supported by numerical and empirical analyses.

\end{abstract}

 

\section{Introduction}\label{Intro}

Estimating  the volatility of asset volatility (hereinafter vol-of-vol) is relevant in many areas of mathematical finance, such as the calibration of stochastic volatility of volatility models (\cite{bnv}, \cite{scm}),  the hedging of portfolios against volatility of volatility risk (\cite{hsst}), the estimation of the leverage effect (\cite{kx}, \cite{aflwy}), and the inference of future returns (\cite{btz}), along with spot volatilities (\cite{mz}).   

The literature offers a number of consistent estimators  for the integrated vol-of-vol. The first estimator to appear was the one proposed by  \cite{bnv}, termed pre-estimated spot-variance based realized variance (PSRV), which is, in fact, simply the realized variance of the unobservable spot variance, computed using estimates of the latter. Later, \cite{v} derived two sophisticated versions of the simple PSRV: one that allows for a central limit theorem with the optimal rate of convergence, but also for negative values, and another that preserves positivity at the expense of a slower rate of convergence. Note that the simple PSRV and its sophisticated versions are consistent when the price and volatility processes are continuous semimartingales, in the absence of microstructure noise contaminations. Further, Fourier-based estimators of the integrated vol-of-vol were introduced by \cite{scm} and \cite{ct}. In particular, the estimator by \cite{scm} is asymptotically unbiased in the presence of market microstructure noise, while the estimator by \cite{ct} allows for a central limit theorem in the presence of jumps in the price and volatility processes.

The numerical studies in \cite{aflwy} and \cite{scm} show that both realized and Fourier-based integrated vol-of-vol estimators may carry a substantial finite-sample bias unless the selection of the tuning parameters involved in their computation is carefully optimized. However, this is a rather unexplored issue, which we aim to explore. To do so, we focus on the simple PSRV, since it represents the most intuitive and easy-to-implement  vol-of-vol estimator.  Furthermore, asymptotically-optimal estimators do not necessarily guarantee the best finite-sample performance, as pointed out in the extensive study by \cite{go} on integrated volatility estimators and confirmed for integrated vol-of-vol estimators by the numerical studies in \cite{aflwy} and \cite{scm}. Thus, there is no reason to expect a priori that the simple PSRV would show worse finite-sample performance than its sophisticated version with optimal rate of convergence.

As mentioned, the PSRV is the realized volatility of the unobservable spot volatility process, computed from discrete estimates of the latter. In other words, the PSRV is the sum of the squared increments of estimates of the unobservable spot volatility on a discrete grid.  These estimates are obtained as local averages of the price realized volatility. Formally, the locally averaged realized variance and the PSRV are defined as follows.

\begin{definition}{\textbf{Locally averaged realized variance}}\label{larv}
 
Suppose that the log-price process $p$ is observable on an equally-spaced grid of mesh size $\delta_N$, with $\delta_N \to 0$ as $N \to \infty$. Also, let $k_N=O(\delta_N^b)$, $b \in (-1,0),$ be a sequence of positive integers such that  $k_N \to \infty$ and define the local window $W_N:=k_N\delta_N$ such that  $W_N \to 0$ as $N \to \infty$. The locally averaged realized variance at time $t$ is defined as

$$\hat \nu_N(t) := \displaystyle \frac{1}{k_N \delta_N} \sum_{j=1}^{k_N} \Big[ p(\lfloor t/ \delta_N\rfloor \delta_N - k_N \delta_N + j\delta_N) - p(\lfloor t/ \delta_N\rfloor \delta_N - k_N \delta_N + (j-1) \delta_N)       \Big]^2,$$

\vspace{0.5cm}
\noindent where $\lfloor\,\cdot\,\rfloor$ denotes the floor function.

\end{definition}

\begin{definition}{\textbf{Pre-estimated spot-variance based realized variance}}\label{psrv}

Suppose that the log-price process $p$ is observable on an equally-spaced grid of mesh size $\delta_N$, with $\delta_N \to 0$ as $N \to \infty$. The pre-estimated spot-variance based realized variance (PSRV) on the interval $[\tau, \tau + h]$ is defined as

$$PSRV_{[\tau, \tau + h],N} := \sum_{i=1}^{\lfloor h/\Delta_N \rfloor} \Big[ \hat \nu_N(\tau+i \Delta_N) - \hat \nu_N(\tau+(i-1) \Delta_N) \Big]^2,$$

\noindent where:

\begin{itemize}
\item[-] $\hat\nu_N(\cdot) $ is the locally averaged realized variance in Definition \ref{larv}, with $k_N=O(\delta_N^b)$, $b \in (-1,0)$;

\item[-] $\Delta_N= O(\delta_N^c)$, $c \in (0,1)$, is the locally averaged realized variance sampling frequency.

\end{itemize}   
\end{definition}

The following propositions summarize the asymptotic properties of the locally averaged realized variance and the PSRV. For further details, see Chapter 8 in\cite{aj}.

\begin{proposition}
Let the log-price process $p$ be a continuous semimartingale and let the process $\nu$ denote its instantaneous volatility. Then $\hat \nu_N(t)$ is a consistent local estimator of $\nu(t)$ as $N \to \infty$. 
\end{proposition}
\begin{proof}  See Chapter 8.1 in \cite{aj}.
 \end{proof}

\begin{proposition}\label{consist}
Let the log-price process $p$ and the spot volatility process $\nu$ be continuous semimartingales. Then the PSRV is a consistent estimator of the quadratic variation of the volatility process $\langle \nu, \nu \rangle_{[\tau,\tau +h]}$  if  $b \in (-1/2,0)$ and $ c \in (0, -b/2)$. 
\end{proposition}
\begin{proof}
See Proposition 8 in \cite{bnv}.\end{proof}

\begin{remark} Note that the requirements for rates $b$ and $c$ that guarantee consistency imply that $\frac{W_N}{\Delta_N} \to 0$ as $N \to \infty$. Indeed, as one can easily verify,
$-1/2<b<0$ and $0<c<-b/2$ imply $ c < 1+b, $ which, in turn, implies $ \frac{W_N}{\Delta_N} \to 0$ as $N \to \infty$.
\end{remark}\label{rem1}

 In practical applications, when computing PSRV values, one has to select the spot volatility estimation grid. Moreover, since the spot volatility is estimated as an average of the price realized volatility over a local window, the length of the latter must also be selected. More specifically, the figure below details the different quantities involved in the computation of the PSRV: the time horizon $h$; the log-price sampling frequency $\delta_N:=\frac{h}{N}$; the spot volatility sampling frequency $\Delta_N:=\lambda_N\delta_N$, $\lambda_N= min(N, \lceil \lambda \delta_N^{c-1}\rceil),$ $\lambda >0$, $c \in (0,1)$; the size of the local window to estimate the spot volatility $W_N=k_N\delta_N$, $k_N= \lceil \kappa \delta_N^b \rceil$, $\kappa >0$, $b \in (-1,0)$; and the spot volatility estimates $\hat\nu(s)$, $s=\tau+j\Delta_N$, $j=0,1,...,\lfloor h / \Delta_N \rfloor$. Note that $\lceil\,\cdot\,\rceil$ denotes the ceiling function.  \\
 
 \begin{figure}[H]\label{sketch}

\vspace{0.1cm}
\begin{center}
\begin{tikzpicture}
  \draw[line width=0.6pt, ->, >=latex'](0,0)--(11,0);
         \foreach \x/\xtext in {2/$ \hat\nu(\tau)$,4/$\hat\nu(\tau+\Delta_N)$,6/$\hat\nu(\tau+2\Delta_N)$,10/$\hat\nu(\tau+h)$}
      \draw(\x,-5pt)--(\x,5pt) node[above] {\xtext};      
     \draw[decorate, yshift=-5ex]  (2.75,0) -- node[below=0.4ex] {$  \delta_N$}  (3,0);
    \draw[decorate, yshift=-5ex]  (4,0) -- node[below=0.4ex] {$\Delta_N=\lambda_N\delta_N$}  (6,0);
     \draw[decorate, yshift=-5ex] (2,0) -- node[below=0.4ex] {$W_N=k_N\delta_N$} (0.75,0);
     \draw[decorate, yshift=-4ex] ;
     \foreach \x in {0.75,1,...,2}  \draw (\x,0) -- (\x,-3pt);
    \foreach \x in {2.75,3,...,4}  \draw (\x,0) -- (\x,-3pt);
        \foreach \x in {4.75,5,...,6}  \draw (\x,0) -- (\x,-3pt);
     \foreach \x in {8.75,9,...,10} \draw (\x,0) -- (\x,-3pt);
 \node [below] at (7.75,0) {$...$};
  \node [below] at (11,0) {$time$};
\end{tikzpicture}
\end{center}
   
   \caption{ Graphic representation of the quantities involved in the computation of the PSRV on the interval $[\tau, \tau+h]$.}  
    \end{figure}
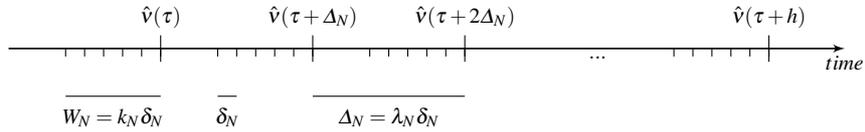

   As a consequence, for given values of the asymptotic rates $b$ and $c$, the finite-sample performance of the PSRV (i.e., the performance of the PSRV for a fixed $N$) depends on the selection of two tuning parameters: $\lambda$, which determines the mesh of the spot volatility estimation grid and $\kappa$, which determines the length of the local window used to estimate the spot volatility.

With regard to the selection of $\kappa$, note that the efficient computation of the spot volatility in finite samples may require the selection of fairly long local windows (see, e.g., \cite{LeeM}, \cite{afl} and \cite{ZB}). This in turn suggests that the finite-sample efficient implementation of the PSRV over a given period (e.g., one day) may require the use of price observations from the previous period(s) (e.g.,  day(s)). At the same time, this might imply that it is optimal to allow consecutive local windows to overlap in finite samples, that is, $W_N>\Delta_N$ for $N$ finite. This aspect is confirmed by the numerical study in \cite{scm}, which shows that it is optimal to select $\kappa$ such that  $W_N>\Delta_N$ in finite samples. The aim of this paper is to gain insight into the bias-reducing effect due to window-overlapping from an analytic perspective. To do so, we follow an approach inspired by the one used in \cite{afl} to solve the ``leverage effect puzzle". 

The ``leverage effect puzzle" pertains to the absence of correlation between log-price and (estimated) volatility changes at high-frequencies, observed in empirical samples.  \cite{afl} solve this puzzle by showing analytically that a substantial bias masks the presence of correlation unless log-price and volatility estimates changes are computed on a suitably sparse grid. The aim is not to solve the problem of the efficient non-parametric estimation of the leverage at high-frequencies, but rather to obtain insight into the puzzle by solving it in a widely used parametric setting that allows for fully explicit computations. 
This paper is written in the same spirit. In fact, we do not address the general problem of the efficient non-parametric vol-of-vol estimation from high-frequency prices, but, rather, our aim is to obtain insight from an analytic perspective into why the PSRV, the simplest and most natural vol-of-vol estimator, is plagued by a large bias in finite samples and investigate the role of window-overlapping as a tool for reducing such large bias.

\section{Outline of the paper}
To achieve this aim, we proceed as follows.
 In Section \ref{mot} we perform a preliminary numerical exercise that uncovers the crucial role of the local-window  parameter $\kappa$ in determining the finite-sample performance of the PSRV and, at the same time, shows that the latter is basically insensitive to the selection of the grid parameter $\lambda$. In particular, it is evident from simulations that the PSRV finite-sample bias is optimized by selecting $\kappa$ such that consecutive local windows to estimate the spot volatility overlap. Numerical results of Section \ref{mot} confirm those by \cite{scm} and motivate the analytic study of Section \ref{anres}.

In Section \ref{anres}, we address the problem of the optimal selection of PSRV tuning parameters in finite samples from an analytic perspective. To do so, in the spirit of \cite{afl}, we assume a widely-used parametric form for the data-generating process, which allows us to obtain the full explicit PSRV finite-sample bias expression. Specifically, we assume the price to be a continuous semimartingale and the volatility to be a CIR process (see \cite{cir}). In general, independently of the parametric assumption on the data-generating process, the PSRV finite-sample bias expression differs in case window overlapping is allowed, i.e., when $W_N > \Delta_N$, or not, i.e., $W_N \le \Delta_N$. Consequently, in Section \ref{anres} we study both cases. 

In the no-overlapping case we adopt a conventional approach and isolate the dominant term of the bias as $N\to \infty$, thereby showing that a value of $\kappa$ that annihilates the dominant term of the bias does not always exist and, even when it does exist, its computation would be basically unfeasible in practice, as it depends on the drift parameters of the volatility, which cannot be reliably estimated on a fixed time horizon, due to the fact that their consistent estimation is possible only in the classic  long-sample asymptotics setting (see, e.g., \cite{kankris}). In addition, when the optimal value of $\kappa$ exists, for typical orders of magnitude of the CIR parameters it actually satisfies the no-overlapping constraint only at ultra-high frequencies ($< 1$ second).

In the overlapping case, instead, the natural expansion as $N \to \infty$ is precluded, as the consistency of the PSRV requires that consecutive windows do not  overlap as the number of price observations grows to infinity (see Remark \ref{rem1}). Therefore, in this case we adopt a novel approach and expand sequentially the bias expression as the tuning parameter $\lambda$ and the time horizon $h$ go to zero, based on the fact that, in practical applications,  $h$ and  $\lambda$ are typically very close to zero. This approach yields a dominant term of the bias which is independent of the tuning parameter $\lambda$ and is annihilated by selecting the asymptotic rate of $k_N$ as $b=-1/2$, the asymptotic rate of $\Delta_N$ as $c<1/2$, and the local-window tuning parameter as $\kappa =2\sqrt{\nu(\tau)}{\gamma}^{-1}$, where $\nu(t) $ and $\gamma$ denote, respectively,  the spot variance process at the initial time $\tau$ and the CIR diffusion parameter.  This analytic result shows that, when overlapping is allowed, it is possible to select $\kappa$ such that the bias is effectively optimized in practical applications and supports  the numerical evidence on the bias-reducing effect due to window overlapping collected in Section \ref{mot} and in \cite{scm}. However, this rule to select $\kappa$ is unfeasible unless reliable estimates of $\nu(\tau)$ and $\gamma$ are available. Accordingly, in Appendix B we detail a simple procedure to estimate $\nu(\tau)$ and $\gamma$.

In Section \ref{anres} we also address the problem of the bias-optimal implementation of the PSRV in the more realistic situation where the price process is contaminated by an i.i.d. microstructure noise process at high frequencies. Again, we distinguish between the overlapping and no-overlapping case and derive, for each case, the exact parametric expression of the extra bias term due to microstructure noise. However, in both cases it emerges that this extra term depends not only on some moments of the noise process but also on the drift parameters of the volatility process, which cannot be consistently estimated
over a fixed time horizon (see, e.g., \cite{kankris}). This precludes the possibility of efficiently subtracting the bias due to noise in small samples. As a solution, we suggest  sampling prices on a suitably sparse grid as in the seminal paper by \cite{abdh}, so that the presence of noise becomes
negligible and the bias optimal rule to select the local-window parameter $\kappa$ can still be applied. The efficiency of this solution is
 verified numerically in Section \ref{num} for typical values of the  noise-to-signal ratio.  For completeness, we also analyze the noise bias expression in the no-overlapping case. In particular, we exploit this expression to derive the asymptotic rate of divergence of the PSRV bias as $N$ tends to infinity.

Additionally, as a byproduct of the PSRV bias analysis, in Section \ref{anres}  we quantify the bias reduction following the assumption that the initial value of the volatility process is equal to the long-term volatility parameter, in the case of both the PSRV and the locally averaged realized variance. This is a very common assumption in the literature, typically made in simulation studies where a mean-reverting process drives the spot volatility (see, e.g., among many others, \cite{afl}, \cite{scm}, \cite{v}).

In Section \ref{DimAn} we use a heuristic approach based on dimensional analysis to generalize the rule for the selection of $\kappa$ to the case of a volatility process belonging to the CKLS class (see \cite{ckls}). Specifically, we find that it is optimal, in terms of bias reduction, to select $ \kappa=2\frac{\nu(\tau)}{\sqrt{\xi(\tau)}}$, where $\nu(t)$ is the variance process and $\xi(t)$ is the variance-of-variance process, while $\tau$ is the initial time of the estimation horizon.  Note that in the absence of price and volatility jumps (a condition required for the PSRV to be consistent), the semi-parametric stochastic volatility model where the price is a semimartingale and the volatility is a CKLS process represents a fairly flexible model. In fact, the CKLS framework encompasses a number of widely-used models for financial applications. Indeed, besides the CIR model, which determines the volatility dynamics in the popular Heston model (\cite{heston}) and its generalized version with stochastic leverage  by \cite{veraart2}, the CKLS family includes, e.g., the model by \cite{brschw} and the model by  \cite{cox}, which appear, respectively, in the continuous-time GARCH stochastic volatility model by \cite{nelson} and 3/2 stochastic volatility model by \cite{32}.

  In Section \ref{num}  we perform an extensive numerical study where we test the performance of the feasible rule to select $\kappa$ derived in Section \ref{anres} and generalized in Section \ref{DimAn}. The results confirm that this rule is effective in reducing the PSRV bias.  We underline that this rule does not require the estimation of the drift parameters of the CIR process, which can not be consistently estimated on a fixed time horizon.  Finally, in Section \ref{emp} we illustrate the results of an empirical study, in which we compute PSRV values from high-frequency S\&P 500 prices, selecting $\kappa$ based on the  bias-optimal rule. Section \ref{concl} summarizes our conclusions. Finally, Appendix A contains the proofs and Appendix B illustrates the feasible procedure that we propose to  select $\kappa$ from sample prices.

\section{Preliminary results}\label{mot}

The finite-sample accuracy of the PSRV requires the careful selection of the tuning parameters $\kappa$ and $\lambda$.  In this section we gain some preliminary insight into this issue by performing a numerical study, whose result motivate the analytic investigation of Section \ref{anres}. In particular, we simulate observations from the following data-generating process, where the volatility is a CIR process. Note that this data-generating process is also used for the analytic study in Section \ref{anres}.

 \begin{assumption}\label{Ass2}{\textbf{Data-generating process}}\label{dgp}

For $t \in [0,T]$, $T>0$, the dynamics of the log-price process $p(t)$ and the spot volatility process $\nu(t)$ read:
 
 $$\begin{cases} \displaystyle p(t) = p(0)+ \int_0^t \sqrt{\nu(s)} dW(s)\\ \displaystyle\nu(t) = \nu(0) + \theta \int_0^t  \Big( \alpha -\nu(s)\Big) ds + \gamma\int_0^t \sqrt{\nu(s)} dZ(s)\end{cases},$$

\noindent where $W$ and $Z$ are two Brownian motions on $(\Omega, \mathcal{F}, (\mathcal{F}_t)_{t \ge 0}, P)$, $p(0) \in \mathbb{R}$ is the initial price, and the strictly positive constants $  \nu(0),\theta, \alpha$ and $\gamma $ denote, respectively, the initial volatility and the speed of mean reversion, long-term mean and vol-of-vol parameters. We also assume that $2\alpha \theta > \gamma^2$ to ensure that $\nu(t)$ is a.s. positive $\forall t \in [0,T]$. 

 \end{assumption}
In particular, we simulate one thousand 1-year trajectories of 1-second observations, with a year composed of 252 trading days of 6 hours each. We consider three scenarios determined by the following sets of model parameters: \textit{Set 1}: $(\alpha, \theta,\gamma, \rho,\nu(0))= (0.2,5,0.5,-0.2,0.2)$;  \textit{Set 2}: $(\alpha, \theta,\gamma, \rho,\nu(0))=(0.02,10,0.25, -0.8, 0.03)$;  \textit{Set 3}: $(\alpha, \theta,\gamma, \rho,\nu(0))= (0.2,5,0.5, -0.2, 0.4)$. 
 
  The first set of parameters, \textit{Set 1}, is taken from \cite{scm} and \cite{v} and is used as the baseline scenario.
 The second, \textit{Set 2}, represents the opposite scenario. In fact, the volatility generated by \textit{Set 2} is lower than the volatility generated by \textit{Set 1}, since the long term mean, $\alpha$, and the speed of mean reversion, $\theta$, are, respectively, much lower and much higher than in \textit{Set 1}.
The second scenario is also characterized by a lower volatility of the volatility, which is captured by the parameter $\gamma$ and a more pronounced  leverage effect, which is captured by the correlation parameter $\rho$. The third set of parameters, \textit{Set 3}, differs from the first only in that the initial value of the volatility,  $\nu(0)$, is twice the long term volatility, $\alpha$. In this regard, note that if the initial volatility $\nu(0)$ is equal to $\alpha$, the spot volatility has a constant unconditional mean over time under Assumption \ref{Ass2} (see Appendix A in \cite{bz}). Setting $\nu(0)=\alpha$ is a simplifying assumption typically adopted in numerical studies where a mean-reverting volatility process is used (see, e.g., among many others, \cite{afl}, \cite{scm}, \cite{v}).

We estimate daily values of the PSRV in these three scenarios from simulated prices sampled with a 1-minute frequency. For the estimation, we set $b=-1/2$ and $c=1/4$ \footnote{This choice of $b$ and $c$ satisfies the constraints for asymptotic unbiasedness (see Theorem \ref{th1}). Moreover, note that the selection $b=-1/2$ is also performed in the numerical exercises of \cite{aflwy} and \cite{scm}.} and study the sensitivity of the bias to different values of $\kappa$ and $\lambda$. With respect to $\lambda$, we consider values in the set $(0.0002, 0.0004, 0.0006, 0.0010, 0.0019, 0.0029, 0.0057)$, which correspond to $\Delta_N$ equal to $1,2,3,5,10,15,30$ minutes, respectively, thereby preserving the high-frequency nature of the estimator.  As for $\kappa$, 
we consider values in the set $(0.017,    0.033,    0.05,    0.1,    0.2,    0.4,    0.5,    1,    1.5,    2,    2.5, 3)$, which correspond to $W_N$ equal to (approximately) $5,    10,    15, 30,    60, 120,   150,   300,   450,   600,   750,   900 $ minutes, respectively. Overall, these sets of values for $\lambda$ and $\kappa$ allow us to consider both cases when window overlapping occurs, that is when $W_N > \Delta_N$, and cases when it does not occur, that is when $W_N \le \Delta_N$. Figure \ref{fig1} summarizes the results of the numerical exercise for values of $\kappa$ that lead to a relative bias smaller than an absolute value of $1$. This happens for $\kappa=1.5, 2, 2.5, 3$. Instead, for $\kappa$ smaller than $1.5$, the relative bias rapidly explodes for all values of $\lambda$ considered, as shown in Figure \ref{fig1b}.

\begin{figure}[h!]
\includegraphics[width=\linewidth]{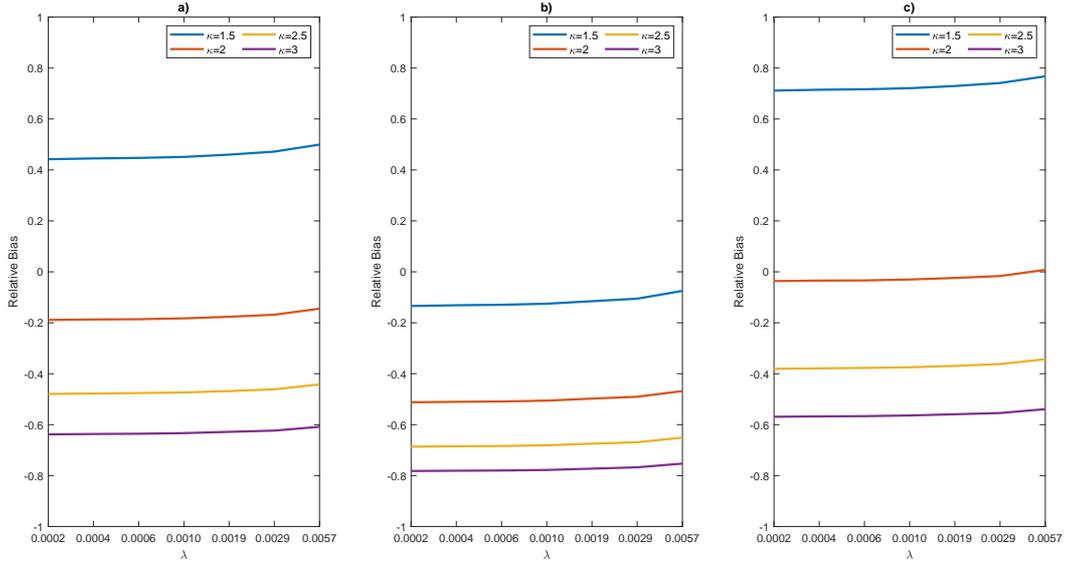}
\caption{Daily PSRV finite-sample relative bias as a function of $\lambda$ for values of $\kappa$ $\in (1.5,2,2.5,3)$ and $\delta_N$ = 1 minute, $b=-1/2$, $c=1/4$. The values of $\lambda$ on the $x$-axis correspond to $\Delta_N$ equal to $j\delta_N$, for $j=1,2,3,5,10,15,30$. The panels refer to the following parameter sets: a)  \textit{Set 1}: $(\alpha, \theta,\gamma, \rho,\nu(0))=(0.2,5,0.5,-0.2,0.2)$; b)  \textit{Set 2}: $(\alpha, \theta,\gamma, \rho,\nu(0))=(0.03,10,0.25,-0.8,0.03)$; and  c) \textit{Set 3}:  $(\alpha, \theta,\gamma, \rho,\nu(0))=(0.2,5,0.5,-0.2,0.4)$.} \label{fig1}
\end{figure}

\begin{figure}[h!]
\includegraphics[width=\linewidth]{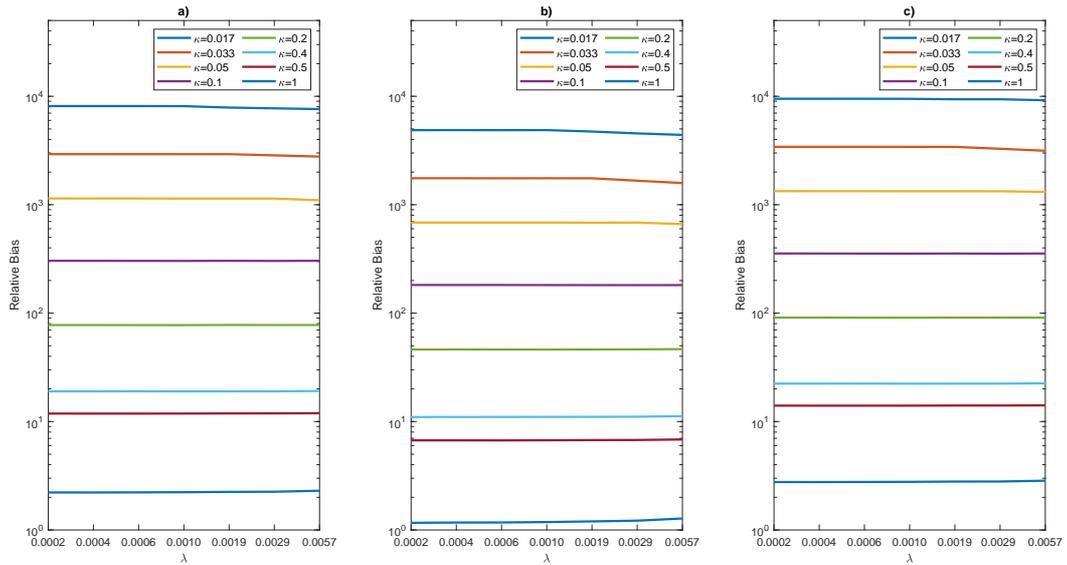}
\caption{Daily PSRV finite-sample relative bias as a function of $\lambda$ for values of $\kappa$ $\in (0.017, 0.033, 0.05, 0.1, 0.2, 0.4, 0.5, 1)$ and $\delta_N$ = 1 minute, $b=-1/2$, $c=1/4$. The values of $\lambda$ on the $x$-axis correspond to $\Delta_N$ equal to $j\delta_N$, for $j=1,2,3,5,10,15,30$. The panels refer to the following parameter sets: a)  \textit{Set 1}: $(\alpha, \theta,\gamma, \rho,\nu(0))=(0.2,5,0.5,-0.2,0.2)$; b)  \textit{Set 2}: $(\alpha, \theta,\gamma, \rho,\nu(0))=(0.03,10,0.25,-0.8,0.03)$; and  c) \textit{Set 3}:  $(\alpha, \theta,\gamma, \rho,\nu(0))=(0.2,5,0.5,-0.2,0.4)$. The y-axis is expressed in log-scale.} \label{fig1b}
\end{figure}

As one can easily verify, the values of $\kappa$ in Figure \ref{fig1} imply that local windows for estimating the spot volatility overlap, for all values of $\lambda$ considered. Consequently, Figure \ref{fig1} tells us that window overlapping is crucial in order to optimize the relative bias of the PSRV  even when {$\Delta_N >> \delta_N$}. This confirms the numerical results in \cite{scm}. Furthermore,  one can also easily check that the combinations of $\lambda$ and $\kappa$ such that overlapping does not occur are all included in Figure \ref{fig1b}, where the relative bias is larger than $1$ and rapidly increases as $\kappa$ becomes smaller, for any $\lambda$  considered, reaching the order magnitude  $10^3$ when $W_N$ equals $5$ minutes.

Moreover, focusing on Figure \ref{fig1}, it is worth noting that the bias-optimal selection of $\kappa$ is strongly dependent on the parameters of the data-generating process. In fact, the same value of $\kappa$ may lead to very different values of the bias in the three scenarios considered: for instance, the selection $\kappa=2$ leads to a relative bias of approximately  $-20\%$ in scenario 1, $- 50\%$ in scenario 2 and $-3\%$ in scenario 3.
At the same time, Figure \ref{fig1} also tells that the bias is not very sensitive to the selection of $\lambda$. Finally, Panel a) of Figure \ref{fig1} suggests that, for all values of $\lambda$ considered, the bias-optimal value of $\kappa$ is close to 2 in scenario 1. This indication is in line with the numerical findings by \cite{scm}, where, based on the same parameter set as in scenario 1, the optimal value of $\kappa$ is found to be approximately equal to 2.
 
In sum, our preliminary numerical study shows not only that allowing for window overlapping is critical to avoid obtaining highly biased vol-of-vol estimates, but also that the selection of $\kappa$ is crucial {for optimizing} the PSRV finite-sample bias and, in particular, it is critical to uncover the dependence between the bias-optimal value of $\kappa$ and the parameters of the data-generating process. Gaining a more in-depth understanding of these numerical findings is what motivates our analytic study in the next section.

 \section{Analytic results}\label{anres}
  
 In this section we analyze the  PSRV  finite-sample bias  in a parametric setting, namely assuming that the volatility is a CIR process, so that a fully explicit formula of the latter can be obtained. We treat the overlapping case (i.e., the case when  $W_N > \Delta_N$) and the no-overlapping case (i.e., $W_N\le\Delta_N$) separately  as, in general, the finite-sample bias expression differs in the two cases, independently of the parametric model used. Lemma \ref{explbias}  details the explicit expression of the PSRV bias for $N$ fixed.

\begin{lemma}\label{explbias}
Let Assumption \ref{Ass2} hold, with $Z$ independent of $W$. Further, let $N$ be fixed. If $W_N \le \Delta_N$, the bias of the PSRV in Definition \ref{psrv} reads

 \begin{equation}\label{truebias}
 E\Big[PSRV_{[\tau,\tau+h],N}- \langle\nu,\nu\rangle_{[\tau,\tau+h]} \Big] = \displaystyle \gamma^2 \alpha h (A_N-1) +  \gamma^2 \Big( E[\nu(\tau)] - \alpha \Big) \frac{1-e^{-\theta h}}{\theta} (B_N-1) + C_ N. 
 \end{equation}

\noindent Instead, if $W_N > \Delta_N$, it reads

\begin{equation}\label{truebiasoverl}
E\Big[PSRV_{[\tau,\tau+h],N}- \langle\nu,\nu\rangle_{[\tau,\tau+h]} \Big] = \displaystyle \gamma^2 \alpha h (A_N-1) +  \gamma^2 \Big( E[\nu(\tau)] - \alpha \Big) \frac{1-e^{-\theta h}}{\theta} (B_N-1) + C_N +O_N.
\end{equation}

The parametric expressions of $A_N$, $B_N$, $C_N$ and $O_N$ are rather cumbersome and thus are reported in the Appendix (see, respectively, equations (\ref{AN}), (\ref{BN}), (\ref{CN}) and (\ref{ON}) in the proof to Lemma \ref{explbias}).
\end{lemma} 

\begin{proof}
See Appendix A.
\end{proof}

\begin{remark}
The bias in the case $W_N \le \Delta_N$ differs from that in the case $W_N > \Delta_N $ for the presence of the extra term $O_N$, which appears due to the fact that the parametric expression of  $E[RV (\tau+i\Delta_N, k_N\delta_N)RV (\tau+i\Delta_N-\Delta_N, k_N\delta_N)]$ differs in the two cases. See the proof of Lemma \ref{explbias} for the definition of the quantity $RV$, which is only used in the Appendix, and further details.
\end{remark}

\begin{remark}
The explicit bias expression in Lemma \ref{explbias} is derived under the simplifying assumption that $Z$ is independent of $W$, which rules out leverage effects.  The sensitivity of the PSRV bias to the presence of leverage effects is studied numerically in Section \ref{num}, where simulations suggest that such effects are a negligible source of finite-sample bias, thereby preserving the practical relevance of the results derived in this section. In the literature, numerical and empirical studies of the impact of leverage effects on analytical results derived under the no-leverage assumption are found, e.g., in \cite{bnvj}, \cite{csda} and references therein.
\end{remark}

 In the next subsections we investigate the existence of a rule to select the tuning parameters $\kappa$ and $\lambda$ in both the cases $W_N > \Delta_N$ and $W_N \le \Delta_N$. To do so, we first isolate the leading term of the bias in each case and then verify whether the latter can be canceled by an $ad$ $hoc$ selection of tuning parameters. We address overlapping case first, as it is the one relevant for practical applications, based on the results of the simulation studies in Section \ref{mot} and  in \cite{scm}.

\subsection{The relevant case for practical applications: $W_N > \Delta_N$}  \label{OVER}

When $W_N > \Delta_N$, the natural expansion of the bias as the number of sampled price observations $N$  tends to infinity is precluded, because the consistency of the PSRV requires that $\frac{W_N}{\Delta_N}\to 0$ as $N \to \infty$. Thus, we determine the leading term of the bias through an alternative asymptotic expansion, which exploits some natural, non-restrictive constraints on the magnitude of the tuning parameter $\lambda$ and the time horizon $h$. Specifically, we first regard the bias in equation (\ref{truebiasoverl}) as a function of $\lambda$ and we perform its Taylor expansion with base point  $\lambda=0$. Then, regarding each term of this expansion as a function of $h$, we perform their Taylor expansions with base point $h=0$. The choice of the base point $\lambda=0$ is supported by the fact that  the largest feasible values of $\lambda$ {are} very small, e.g.,  on the order  of  $10^{-3}$ when $c<1/2$ and $\delta_N$ is equal to one minute (see Figure \ref{fig1} for the case $c=1/4$).  Note that a  value of $\lambda$ is  feasible if it satisfies  $\Delta_N:=\lambda\delta_N^c < h$. The choice of base point $h=0$ is instead supported by the fact that in the literature on high-frequency econometrics, the typical time horizon used to estimate the integrated quantities is one trading day, i.e., $h=1/252\approx 4 \cdot 10^{-3}$. The order of this sequential expansion is rather natural: intuitively, we first take the limit $\lambda \to$ 0 to approximate the integral of the vol-of-vol in an infill-asymptotics sense, then take the limit $h \to$ 0 to localize the estimate of the integral near the initial time $\tau$. This approach leads to the following result.

\begin{theorem}\label{th4}
 Let Assumption \ref{Ass2} hold, with $Z$ independent of $W$. Further, let $W_N>\Delta_N$. Then, for  $N$ fixed, as $\lambda\to 0, h\to 0$ 
 
\begin{equation}\label{overlapuncond}
E\Big[PSRV_{[\tau,\tau+h],N} - \langle\nu,\nu\rangle_{[\tau,\tau+h]}  \Big] =\begin{cases} \displaystyle  \Bigg(  \frac{ 4E[\nu(\tau)]^2}{ {\kappa}^2\delta_N^{1+2b}} -\gamma^2 E[\nu(\tau)]    \Bigg)h  + O(h^{1-b}) +O(\lambda)  \quad \textit{if} \quad b\ge-1/2, c<-b \\   \displaystyle -\gamma^2 E[\nu(\tau)]  h  + O(h^{-2b}) +O(\lambda)  \quad \textit{if} \quad b<-1/2, c<1+b \\   
  \end{cases}.
\end{equation}

\noindent Moreover, let $(\mathcal{F}^{\nu}_t)_{t\ge0}$ be the natural filtration associated with the process $\nu$. Then, for $N$ fixed, as $\lambda\to 0, h\to 0$

\begin{equation}\label{overlapcond}
E\Big[PSRV_{[\tau,\tau+h],N} - \langle\nu,\nu\rangle_{[\tau,\tau+h]}  | \mathcal{F}^{\nu}_\tau\Big] = \begin{cases} \displaystyle  \Bigg(    \frac{4\nu(\tau)^{2 }}{ {\kappa}^2\delta_N^{1+2b}}  -   \gamma^2\nu(\tau)      \Bigg)h  + O(h^{1-b}) +O(\lambda)  \quad \textit{if} \quad b\ge-1/2, c<-b \\ \displaystyle  -\gamma^2\nu(\tau)   h  + O(h^{-2b}) +O(\lambda)  \quad \textit{if} \quad b<-1/2, c<1+b \\   
  \end{cases}.
\end{equation}

\end{theorem}
 
\begin{proof}
See Appendix A.  
\end{proof}

\begin{remark}
The expansion in Theorem \ref{th4} is performed under the asymptotic constraints on rates $b$ and $c$ that ensure the asymptotic unbiasedness of the PSRV under Assumption \ref{Ass2} (see Theorem \ref{th1}).
\end{remark} 

\begin{remark}
The conditional bias expansion in equation (\ref{overlapcond}) allows the dominant term of the bias  {to be expressed} in terms of $\nu(\tau)$ and $\gamma$, two quantities  {that} can be consistently estimated over a fixed time horizon. This is crucial for the existence of a feasible procedure to select $\kappa$, as detailed below. Instead, the unconditional expression in equation (\ref{overlapuncond}) depends on $E[\nu(\tau)]$, whose parametric expression in turn depends on the drift parameters of the volatility and thus cannot be consistently estimated over a fixed time horizon (see, e.g., \cite{kankris}). In particular, it holds
$E[\nu(\tau)]=(\nu(0)-\alpha)e^{-\theta \tau} +\alpha$ (see equation (4) in Section 2.2.1 of \cite{bz}).
\end{remark}

Figure \ref{fig2} compares the true finite-sample bias of the daily PSRV in equation \ref{truebiasoverl} with the dominant term of the expansion in equation \ref{overlapuncond}  as functions of the tuning parameter $\kappa$. Specifically, the panels refer to the three parameter sets already used in Section \ref{mot}, that is \textit{Set 1} (panel a)), \textit{Set 2} (panel b)), and \textit{Set 3} (panel c)). Note that we have set $b=-1/2$, $c=1/4$, $\lambda=0.0006$, $h=1/252$, and $N=360$. 
 The corresponding $\delta_N$ and $\Delta_N$ are equal to 1 minute and  (approximately)  3 minutes, as we consider 6-hr trading days. The approximation of the true bias with the dominant term of the expansion is very accurate.

\begin{figure}[h]
\includegraphics[width=\linewidth]{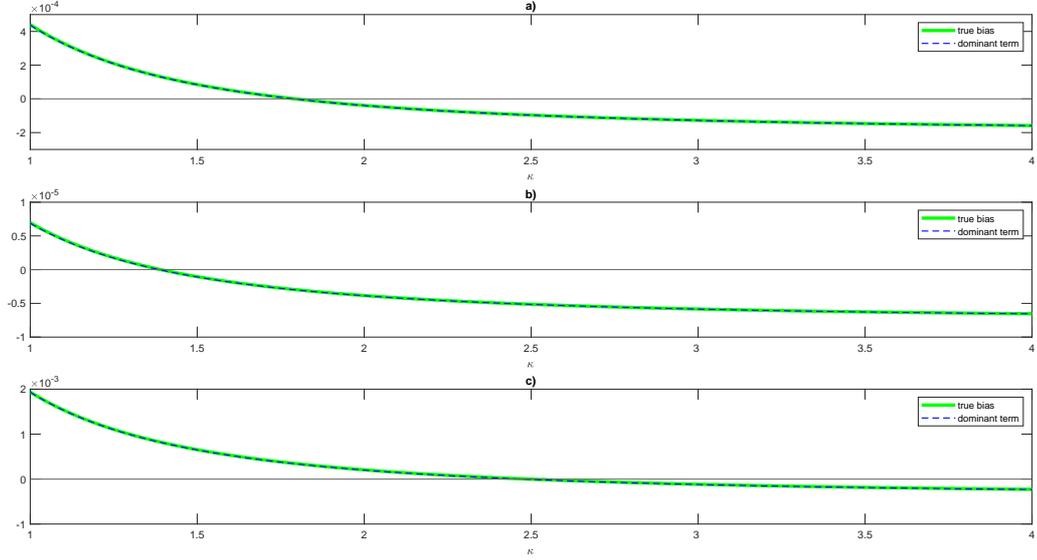}
\caption{Comparison between the true finite-sample bias of the daily PSRV in equation \ref{truebiasoverl} and the dominant term of the expansion in Theorem \ref{th4}  as functions of $\kappa$ for $b=-1/2$, $c=1/4$, 
$\lambda=0.0006$, and $N=360$. Panel a) refers to the parameter set $(\alpha, \theta,\gamma, \nu(0))
=(0.2,5,0.5,0.2)$; panel b) to $(\alpha, \theta,\gamma, \nu(0))=(0.03,10,0.25,0.03)$, and panel c) to $(\alpha, \theta,\gamma,  \nu(0))=(0.2,5,0.5,0.4)$. For panel c) we consider $\tau = 5$ days, while  {the bias terms  in panels a) and b) are independent} of $\tau$.} \label{fig2}
\end{figure} 
  
 Based on the conditional bias expansion in equation (\ref{overlapcond}), we make the following considerations on the optimal selection of tuning parameters in finite samples. 
 First, we note that the dominant term of the bias can be annihilated simply by suitably selecting $\kappa$ for any feasible value of $\lambda$ when $b\ge -1/2$, $c<-b$. Instead, when $b<-1/2$, $c<1+b$, the dominant term of the bias is independent of $\kappa$ and $\lambda$. Specifically, when $b\ge -1/2$, $c<-b$, the suitable selection is
\begin{equation}\label{kappaconB}  
\kappa=   \frac { 2\sqrt{\nu(\tau) }}{\hat{\gamma} \delta_N^{b+1/2}}.
\end{equation}
 
 { However, since $\kappa$ is a tuning parameter, it is not allowed to depend on $N$. Therefore, the only admissible choice is $b=-1/2$ and $c<1/2$,  so that the suitable selection becomes}  
 
 \begin{equation}\label{kopt}
\kappa= \kappa^*:=  \displaystyle \frac{2 \sqrt{\nu(\tau)}}{\gamma}. 
 \end{equation}
 
Further, if $\nu(0)=\alpha$, then $E[\nu(\tau)]=\alpha$ (see equation (4) in Section 2.2.1 of \cite{bz}) and thus, based on equation (\ref{overlapuncond}), it is immediate to see that the bias-optimal value of $\kappa$ reduces to $ \displaystyle \frac{2 \sqrt{\alpha}}{\gamma}.$ Interestingly,  this analytic result  supports the optimal selections of $b$ and $\kappa$ determined numerically in the literature.
Indeed, for the first parameter set in the numerical exercise in Section \ref{mot}, \textit{Set 1}, which is also used in \cite{scm}, $\kappa^*$ is  equal to $1.79$, a value compatible with the numerical result in \cite{scm}, where the optimal $\kappa$ is said to be approximately equal to $2$. 
 Note also that the numerical studies in \cite{aflwy}, \cite{scm} both select $b=-1/2$.  With regard to this selection of $b$, the following remark is in order.

\begin{remark}
The selection $b=-1/2$ does not satisfy the consistency constraint in Proposition \ref{psrv}. However, to achieve consistency, it is sufficient to select $b=-1/2+\epsilon$ and $c=1/4 -\epsilon$, with $\epsilon$ strictly positive but arbitrarily small and, for such selection of $b$,
bearing in mind (\ref{kappaconB}),  the impact of $\epsilon$ on the optimal selection of $\kappa$ will be negligible in finite-sample exercises. Therefore, in finite-sample exercises we can select $\kappa=\kappa^*$ as in (\ref{kopt}).
\end{remark}

Furthermore, the following remark regarding the selection $\kappa=\kappa^*$ is in order.

\begin{remark}
 {The overlapping condition  $W_N>\Delta_N$ implies a constraint on the price grid $\delta_N$. In particular, if $\kappa=\kappa^*$, for $b=-1/2$, $c<1/2$,  $W_N>\Delta_N$ is equivalent to $\delta_N > \delta^*:=\displaystyle \Big( \frac{\kappa^*}{ \lambda}\Big)^{\frac{1}{c-1/2}} $. The threshold $\delta^*$ is very small  for typical orders of magnitude of $\alpha, \theta $ and $\gamma$,  $h$ corresponding to one trading day and any feasible value of $\lambda$.
 For example, for the values of the  parameters in \textit{Set 1} (see Section \ref{mot}), $\lambda=0.0006$ and $c=1/4$ (so that, if $\delta_N=1$ minute, then $\Delta_N:=\lambda\delta_N^c \approx 3$ minutes), we have $\delta^*=7.5 \cdot 10^{-8}$ seconds and thus the constraint $\delta_N>\delta^*$ is largely satisfied at the most commonly  available price sampling frequencies.}  
 \end{remark}

However, Equation (\ref{kopt}) implies that the bias-optimal selection  $\kappa:=\kappa^{*}$ is unfeasible unless reliable estimates of $\nu(\tau)$ and $\gamma$ are available. In Appendix B we detail a simple feasible procedure to obtain $\kappa^{*}$. In a nutshell, the procedure is as follows. First, we estimate $\nu(\tau)$ using the Fourier spot volatility estimator by \cite{mm}. Then we estimate $\gamma$ via a simple indirect inference method.

\subsection{The case $W_N \le \Delta_N$} 
 
The finite-sample bias expression for $W_N \le \Delta_N$ in equation (\ref{truebias}) is the starting point to derive the asymptotic constraints on rates $b$ and $c$ that ensure the asymptotic unbiasedness of the PSRV. In this regard, we obtain the following result, which is based on the asymptotic expansion of the bias in the limit $N \to \infty$.

\begin{theorem}\label{th1}
Let Assumption \ref{Ass2} hold, with $Z$ independent of $W$.  Then, if $b\ge -1/2$ and $c < -b$ or $b<-1/2$ and $c<1+b $, $\frac{W_N}{\Delta_N} \to 0 $ as $N \to \infty$ and the PSRV as given in Definition \ref{psrv} is asymptotically unbiased, i.e.,

 $$ E\Big[PSRV_{[\tau,\tau+h],N} - \langle\nu,\nu\rangle_{[\tau,\tau+h]}  \Big] \to 0  \quad \textit{as} \quad N \to \infty. $$

\noindent In particular, as $N \to \infty$,    
\begin{eqnarray}\label{dom}
 E\Big[PSRV_{[\tau,\tau+h],N} - \langle\nu,\nu\rangle_{[\tau,\tau+h]}  \Big]= a_1 \Delta_N+ a_2 \frac{1}{k_N\Delta_N}+a_3 \frac{k_N \delta_N}{\Delta_N} + o\Big(\Delta_N\Big) + o\Big(\frac{1}{k_N\Delta_N}\Big) + o \Big(\frac{k_N \delta_N}{\Delta_N}\Big),   
\end{eqnarray}

where:
\begin{eqnarray*} 
&&a_1=   - \frac{\theta }{2} \gamma^2\alpha h + \frac{\theta}{2} \gamma^2 (E[\nu(\tau)] -\alpha)\frac{1-e^{-\theta h}}{\theta} + \frac{\theta}{2} (1-e^{-2 \theta h})\Big[(E[\nu(\tau)]-\alpha)^2+\frac{\gamma^2}{\theta}\Big(\frac{\alpha}{2}-E[\nu(\tau)]\Big)\Big]   , \\
&&a_2= \frac{2}{\theta}\gamma^2\alpha h + \frac{4}{\theta}\gamma^2 (E[\nu(\tau)] -\alpha)\frac{1-e^{-\theta h}}{\theta} +\frac{2}{ \theta }(1-e^{-2 \theta h})\Big[(E[\nu(\tau)]-\alpha)^2+\frac{\gamma^2}{\theta}\Big(\frac{\alpha}{2}-E[\nu(\tau)]\Big)\Big]     \\ 
&& \quad\quad +  4 \alpha^2 h +\frac{8 \alpha (E[\nu(\tau)]-\alpha)(1-e^{-\theta h})}{\theta  },\\
&& a_3=-\gamma^2 (E[\nu(\tau)] -\alpha)\frac{1-e^{-\theta h}}{\theta}.   
\end{eqnarray*}

\end{theorem}

\begin{proof}
See Appendix A.
\end{proof}

A bias-optimal rule for the selection of the tuning parameters $\kappa$ and $\lambda$ when $W_N \le \Delta_N$ is given in the following corollary to Theorem \ref{th1}. Unfortunately, this bias-optimal rule is of little interest for practical applications, as explained in Remark \ref{rm2}.

\begin{corollary}
 The leading term of the PSRV finite-sample bias expansion in Eq. (\ref{dom}) can be canceled in the case  $b=-1/2$ and  $c=1/4$, provided that there exists a solution ${(\tilde\kappa , \tilde\lambda ) \in \mathbb{R}_{>0}\times  \mathbb{R}_{>0}}$ to the following system:

$$\begin{cases} \displaystyle a_3 \kappa^2 + a_1 \lambda^2 \kappa + a_2=0   \\  W_N \le \Delta_N  
 \end{cases}.$$

If a solution $(\tilde\kappa , \tilde\lambda )\in \mathbb{R}_{>0}\times  \mathbb{R}_{>0}$ exists, the corresponding bias-optimal selection of $W_N$ and $\Delta_N$ reads $$W_N=\tilde\kappa \delta_N^{1/2}, \ \Delta_N=\tilde\lambda \delta_N^{1/4}.$$
\end{corollary}
\begin{proof}
See Appendix A. 
\end{proof}

\begin{remark}\label{rm2}
 For $b =-1/2$ and $c=1/4$, the no-overlapping condition $W_N \le \Delta_N$ is equivalent to $  \delta_N \le (\lambda/ \kappa)^{4}$. Assuming that a positive solution $\tilde\kappa(\lambda)$ to $a_3 \kappa^2 + a_1 \lambda^2 \kappa + a_2=0$ exists for some $\lambda>0$, we define the ``no-overlapping" threshold for $\delta_N$ as {$\delta^*(\lambda):=  (\lambda / \tilde\kappa(\lambda))^{4}$}. For the three sets of CIR parameters used in the numerical study in Section \ref{mot}, Figure \ref{deltastar} plots the threshold $\delta^*(\lambda)$ as a function of  $\lambda\in (0,\lambda^*]$, where  $\lambda^*$ is the largest admissible value of $\lambda$ such that  $\lambda \delta^*(\lambda)^{1/4}\le h$, i.e., such that $\Delta_N \le h$ when $\delta_N$ is equal to the ``no-overlapping" threshold. Specifically, Figure \ref{deltastar} shows that the sampling frequency corresponding to $\delta^*(\lambda)$ is bounded by a value smaller than, respectively,  $0.02$ (see Panel a)), $0.05$ (see Panel  b)) and $0.125$ seconds (see Panel c)). This suggests that for typical values of the CIR parameters, the system in Corollary 1 may be solved only for ultra-high frequencies. Also, note that the solution (if it exists)  depends on $a_1$, $a_2$ and $a_3$, which in turn depend on the expected initial volatility and all CIR parameters, including the drift parameters, which can not be consistently estimated over a fixed time horizon.
\end{remark}

\begin{figure}[h]
\includegraphics[width=\linewidth]{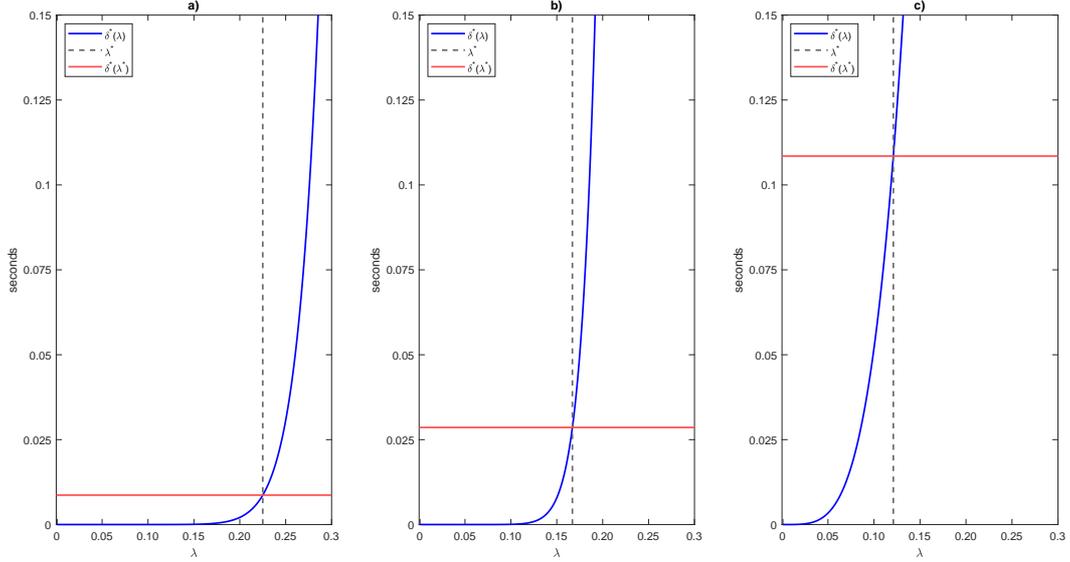}

\caption{ The threshold $\delta^*(\lambda)$ is plotted in black. The dotted gray vertical line corresponds to $\lambda=\lambda^*$. The dotted red horizontal line corresponds to $\delta^{*}(\lambda^*)$. The panels refer to the following sets of parameters: a)  \textit{Set 1}: $(\alpha, \theta,\gamma, \nu(0))=(0.2,5,0.5,0.2)$; b)  \textit{Set 2}: $(\alpha, \theta,\gamma,\nu(0))=(0.03,10,0.25,0.03)$; and  c) \textit{Set 3}:  $(\alpha, \theta,\gamma, \nu(0))=(0.2,5,0.5,0.4)$. {$\delta^*(\lambda)$ is independent of the correlation parameter $\rho$, which therefore does not appear}. For panel c) we consider $\tau = 5$ days, while in panels a) and b) $\delta^*(\lambda) $  is independent of $\tau$. We have assumed $h=1/252$, corresponding to 6 hours (21600 seconds).}  \label{deltastar}
\end{figure}

\subsection{The impact of noise on the bias}
In empirical applications one can only observe the noisy price $\tilde p (t)$, that is, the efficient price contaminated by a noise component that originates from market microstructure frictions, such as bid-ask bounce effects and
price rounding. Here, we assume that the noise component is an i.i.d. process independent of the efficient price process, as in the seminal paper by \cite{roll}. 
For a general discussion of the statistical models of microstructure noise, see \cite{jli}.

 \begin{assumption}\label{Ass3}{\textbf{Data-generating process in the presence of market microstructure  noise}}\label{noise}

The observable price process $\tilde p $ is given by
$$\tilde p (t) = p(t) + \eta(t),$$

\noindent where $p(t)$ represents the efficient price process and evolves according to Assumption \ref{Ass2} while 
$\eta(t)$ is a sequence of i.i.d.\ random variables independent of $p(t)$, such that
 $E[\eta(t)]=0$,  $E[\eta(t)^2]=V_{\eta} < \infty $  and $E[\eta(t)^4]=Q_{\eta} < \infty$ $\forall t$.

 \end{assumption}
 
The presence of noise clearly changes the PSRV bias expression, introducing an extra term, as illustrated in the following lemma. Note that the parametric form of the extra bias term due to the presence of noise is different in the overlapping and no-overlapping cases.

\begin{lemma}\label{explbiasNoise}
Let Assumption \ref{Ass3} hold, with $Z$ independent of $W$, and let $N$ be fixed. Moreover, let $\widetilde{PSRV}_{[\tau,\tau+h],N}$ denote the PSRV in Definition \ref{psrv}, computed from noisy price observations. If $W_N \le \Delta_N$, then

 \begin{equation}\label{truebiasNoise}
 E\Big[\widetilde{PSRV}_{[\tau,\tau+h],N}- \langle\nu,\nu\rangle_{[\tau,\tau+h]} \Big] = \displaystyle \gamma^2 \alpha h (A_N-1) +  \gamma^2 \Big( E[\nu(\tau)] - \alpha \Big) \frac{1-e^{-\theta h}}{\theta} (B_N-1) + C_ N +D_N;
 \end{equation}

\noindent Instead, if $W_N > \Delta_N$, then

\begin{equation}\label{truebiasoverlNoise}
E\Big[\widetilde{PSRV}_{[\tau,\tau+h],N}- \langle\nu,\nu\rangle_{[\tau,\tau+h]} \Big] = \displaystyle \gamma^2 \alpha h (A_N-1) +  \gamma^2 \Big( E[\nu(\tau)] - \alpha \Big) \frac{1-e^{-\theta h}}{\theta} (B_N-1) + C_N +O_N+ D^{*}_N.
\end{equation}

The parametric expressions of $A_N$, $B_N$, $C_N$ and $O_N$ are as in  Lemma \ref{explbias}, while that of the extra term due to the presence of noise $D_N$ (resp., $D^*_N$) in the no-overlapping case (resp., overlapping case) is as follows:

\begin{equation}\label{D_N}
D_N =[4(Q_\eta+V_\eta^2)+16\alpha V_\eta\delta_N]h\frac{1}{k_N\delta_N^2\Delta_N}+\frac{8}{\theta}V_\eta (\alpha -E[\nu(\tau)])(1-e^{-\theta h})\frac{(1+e^{-\theta \Delta_N})(1-e^{-\theta k_N\delta_N})}{(1-e^{-\theta \Delta_N})k_N^2\delta_N^2};
\end{equation}   

\begin{equation}\label{D_Ns}\begin{split}
D^*_N&=[   4\left(Q_\eta+V^2_\eta\right)+ 16\alpha V_\eta \delta_N] h\frac{1}{k_N^2\delta_N^3}+ \frac{8}{\theta}V_\eta(\alpha-E[\nu(\tau)]) (1-e^{-\theta h})\frac{1}{(1-e^{-\theta\Delta_N}) k_{N}^{ 2} \delta_{N}^{ 2}  } \Bigg\{ \frac{(2+  k_{N} )}{2k_N\delta_N}\Bigg[ \frac{(e^{\theta k_N \delta_N - \theta \Delta_N}-1)(k_N\delta_N+\Delta_N)}{k_N\delta_N-\Delta_N}\\ &+(e^{-\theta \Delta_N}-e^{\theta k_N\delta_N})\Bigg]+  \frac{k_N}{2\Delta_N}(1+e^{\theta k_N\delta_N})(1-e^{-\theta\Delta_N})\Bigg\}.\end{split}
\end{equation}

\end{lemma} 

\begin{proof}
See Appendix A.
\end{proof}

\begin{remark}
From the proof of Theorem \ref{th2} in Appendix A, one can easily see that the expressions of $D_N$  is the same for any continuous mean-reverting volatility model, as their computation only depends on the drift of $\nu$ in Assumption \ref{Ass2}. The same holds also for $D^*_N$ in the overlapping case.
\end{remark}

 Ideally, in the overlapping case, if one could efficiently estimate the extra bias due to noise $D^*_N$  and subtract it, then  the bias-optimal rule to select $\kappa$  could still be applied effectively. Unfortunately,   $D^*_N$ can not be consistently estimated over a fixed time horizon, as it depends on the drift parameters of the volatility $\alpha$ and $\theta$, whose consistent estimation can not be achieved on a fixed time horizon\footnote{Note that $Q_\eta$ and $V_\eta$  can, instead, be  estimated consistently for $T$ fixed, see for instance \cite{ZMA}}. As a solution, we suggest to sample prices on a suitably sparse grid, as done for the realized variance in the seminal paper by \cite{abdh}, so that the extra bias term induced by the presence of noise becomes
negligible and the bias optimal rule to select the local-window parameter $\kappa$ can still be applied. The efficiency of this solution is
 verified numerically in Section \ref{num}.

Finally, for completeness, we also study the asymptotic behavior of the additional bias due to noise in the no-overlapping case, $D_N$. More precisely, in the next theorem we derive its rate of divergence as $N \to \infty$.

\begin{theorem}\label{th2}
Let Assumption \ref{Ass3} hold, with $Z$ independent of $W$. Moreover, let $\widetilde{PSRV}_{[\tau,\tau+h],N}$ denote the PSRV in Definition \ref{psrv}, computed from noisy price observations. Then, if either $b \ge -\frac{1}{2}$ and  $c < -b $  or $b < -\frac{1}{2} $ and  $c < b+1 $,  $\frac{W_N}{\Delta_N} \to 0$ as $N \to \infty$ and $\widetilde{PSRV}_{[\tau,\tau+h],N}$ is asymptotically biased, i.e., 

$$ E\Big[\widetilde{PSRV}_{[\tau,\tau+h],N} -\langle\nu,\nu\rangle \Big]_{[\tau,\tau+h],N} \to \infty, \ \ as \ \ N \to \infty, $$

since the bias term $D_N$ in equation (\ref{truebiasNoise}) of Lemma \ref{explbiasNoise} diverges as $N \to \infty.$ In particular, we have

   \begin{equation*}
  k_N\delta_N^2\Delta_N D_N =4(Q_\eta+V^2_\eta)h +O(\delta_N) \quad \textit{and}\quad k_N\delta_N^2\Delta_N \to 0 , \,\,\, \delta_N\to 0  .
\end{equation*}
\end{theorem}

\begin{proof}
See Appendix A.
\end{proof}

 \subsection{The bias-reducing effect of the assumption $\nu(0)=\alpha$}
As mentioned in Section \ref{OVER}, if $\nu(0)=\alpha$, then $E[\nu(\tau)]=\alpha$. Lemmas \ref{explbias} and \ref{explbiasNoise}  quantify the bias reduction ensuing from assuming that $ \nu(0)  = \alpha$. Indeed, this assumption cuts off the entire source of bias $B_N$  and part of the sources of bias $D_N$ (see equation (\ref{D_N})) or $D^*_N$ (see equation (\ref{D_Ns})).  The finite-sample bias reduction ensuing from the assumption $\nu(0)=\alpha$ is not peculiar to the PSRV, though. {In fact}, this simplifying assumption is also beneficial for reducing the finite-sample bias of the locally averaged realized variance, as shown in the next theorem.

\begin{theorem}\label{th3}
Let Assumption \ref{Ass2} hold.  Moreover, let $\hat\nu(\tau)$ denote the locally averaged realized variance in Definition \ref{larv} at time $\tau$. Then, if $b \in (-1,0)$, $\hat\nu(\tau)$ is asymptotically unbiased, i.e.,

\begin{equation*}
 E[\hat\nu(\tau) - \nu(\tau)] =(   \nu(0)  -\alpha)e^{-\theta \tau}\frac{e^{\theta k_N \delta_N}-1- \theta k_N \delta_N}{\theta k_N \delta_N}, 
\end{equation*}
and, as $N \to \infty$, we have

\begin{equation*}
{ E[\hat\nu(\tau) - \nu(\tau)]=\frac{\theta}{2} (  \nu(0) -\alpha)e^{-\theta\,\tau}k_N\delta_N+o(k_N \delta_N), \quad k_N\delta_N\to 0}. 
\end{equation*}

Let Assumption \ref{Ass3} hold. Moreover, let $w(\tau)$ denote the locally averaged realized variance in Definition \ref{larv} at time $\tau$ computed from noisy price observations. Then, $\forall$ $b \in (-1,0)$, $w(\tau)$ is asymptotically biased, i.e.,

\begin{equation*} 
E[w(\tau) - \nu(\tau)] =(  \nu(0)  -\alpha)e^{-\theta \tau}\frac{e^{\theta k_N \delta_N}-1- \theta k_N \delta_N}{\theta k_N \delta_N} + \frac{2V_\eta}{\delta_N}, 
\end{equation*}
and, as $N \to \infty$, we have
\begin{equation*}
{ E[w(\tau) - \nu(\tau)]=\frac{\theta}{2} ( \nu(0) -\alpha)e^{-\theta\,\tau}k_N\delta_N+\frac{2V_\eta}{\delta_N}+o(k_N\delta_N), \quad \quad k_N \delta_N\to 0}. 
\end{equation*}

\end{theorem}
\begin{proof}
See Appendix A.
\end{proof}
This theorem has two interesting implications. First, under Assumption \ref{Ass2}, the locally averaged realized variance is unbiased in finite samples if and only if $ \nu(0)  =\alpha$. Second, under Assumption \ref{Ass3}, if $\alpha >  \nu(0) $, the presence of noise could actually compensate for the negative bias originating from the first term of the bias expression.
This also holds for the  PSRV finite-sample bias, provided that the term $D_N$  (resp., $D^*_N$) in Lemma \ref{explbiasNoise} is of opposite sign with respect to the sum of the other terms in the bias expression.
 
 \vspace{0.5cm}

 \section{Generalization via dimensional analysis }\label{DimAn}

In this section we propose a heuristic approach, based on dimensional analysis, to generalize the rule for the bias-optimal selection of $\kappa$ in equation (\ref{kopt}), derived under the assumption that the volatility is a CIR process,  to the more general case where the volatility follows a process in the CKLS class (see \cite{ckls}). Specifically, the stochastic volatility model we assume as the data-generating process is now as follows.

 \begin{assumption}{\textbf{Data-generating process}}\label{Ass4}

For $t \in [0,T]$, $T>0$, the dynamics of the log-price process $p(t)$ and the spot volatility process $\nu(t)$ follow

$$\begin{cases} \displaystyle p(t) =p(0)+ \int_0^t \mu(s)ds +   \int_0^t \sqrt{\nu(s)} dW(s) \\ \displaystyle \nu(t) = \nu(0) + \theta \int_0^t  \Big( \alpha -\nu(s)\Big) ds + \gamma\int_0^t  \nu(s)^\beta dZ(s) \end{cases}, $$

\noindent where $W$ and $Z$ are two correlated Brownian motions on $(\Omega, \mathcal{F}, (\mathcal{F}_t)_{t \ge 0}, P)$, $\mu(t)$ is a continuous adapted process, $ \beta \ge 1/2$, $p(0) \in \mathbb{R},  \nu(0), \theta, \alpha, \gamma >0$, and   $2\alpha \theta > \gamma^2$ if $\beta=1/2$.

\end{assumption}

 The stochastic volatility model in Assumption \ref{Ass4} is quite flexible to reproduce empirical prices behaviour in the absence of price and volatility jumps. In fact it
incorporates a number of widely-used stochastic volatility models with continuous price and volatility paths as special cases. For example, if $\beta = 1/2$, one obtains the model by \cite{heston}; if $\beta=1$ one finds the continuous-time Garch model by \cite{nelson}; if $\beta=3/2$, one gets the 3/2 model by \cite{32}. Further, by allowing for a stochastic correlation between $W$ and $Z$, Assumption \ref{Ass4} includes also the generalized Heston model with stochastic leverage introduced by \cite{veraart2}. Finally, note that  Assumption \ref{Ass4} also includes a price drift. The numerical study in Section \ref{num}  confirms that the impact of the latter on the PSRV finite-sample bias is negligible.

We now use dimensional analysis  to heuristically derive a rule for the bias-optimal selection of $\kappa$ under Assumption \ref{Ass4}.
We test the efficacy of this rule in the numerical study of Section \ref{num}, with overwhelming results. Note that dimensional analysis is typically used in physics and engineering to make an educated guess about the solution to a problem without performing a full analytic study (see, e.g., \cite{kyle}, \cite{krishnam}).

The basic concept of dimensional analysis is that one can only add quantities with the same units\footnote{{Dimensional analysis is also called a unit-factor method or a factor-label method, since a conversion factor is used to evaluate the units.}}.  Accordingly, when applying dimensional analysis, the first step entails identifying the units of the quantities appearing in the equations being studied. In this specific analysis, we start with the units of  the quantities  appearing in the model given in Assumption \ref{Ass4}. Let $dim[q]$ denote the unit/dimension of the quantity $q$. The log-return $dp(t)$, $t>0$, is a dimensionless quantity (i.e.,  a pure number) since it is the logarithm of a ratio of prices (the ratio of quantities with the same units is dimensionless). Instead, the quadratic variation of the Wiener processes $W$ and $Z$  has the dimension of $time$ since $W$ and $Z$ are  continuous random walks. As a consequence, we have  $dim[dW(t)]=dim[dZ(t)]=time^{1/2}$ (see, for example, \cite{WILMOTT} or the square-root-of-time rule in \cite{sqrtTime}). Now consider the dynamics of the log-price, bearing in mind that we cannot add or subtract quantities with different measurement units.  The dimension of the left-hand side must then be equal to those of the addenda on the right-hand side, thereby implying that $dim[\mu] =1/time$ and
$\dim[\nu(t)]=1/time$. Thus, from the dynamics of $\nu(t)$, we have $dim[\alpha]=1/time$, $dim[\theta]=1/time$
and $dim[\gamma\nu(t)^\beta \,dZ(t)]=1/time$. The latter implies $dim[\gamma]dim[\nu(t)^\beta]dim[dZ(t)]=1/time$.  Therefore, bearing in mind that 
$dim[\nu(t)^\beta]dim[dZ(t)]=1/time^{\beta-1/2}$, we obtain $dim[\gamma] =1/time^{-\beta+3/2}$.

Now, without loss of generality, let $ \nu(0) =\alpha$ and consider the dominant term in the expansion of Theorem \ref{th4}, i.e., the term 
$$\Big(\frac{4 \alpha^2}{\kappa^2\delta_N^{1+2b}} -\gamma^2\alpha\Big)h.$$ 
Since the dominant term of the PSRV bias must clearly have the same dimension as the expected quadratic variation of $\nu$ over any generic interval of length $h$, i.e., $\gamma^2 \alpha h$, we have

$$dim\Big[\Big(\frac{4 \alpha^2}{\kappa^2\delta_N^{1+2b}} -\gamma^2\alpha\Big)h\Big] =dim[ \gamma^2 \alpha h] = 1/time{^2},$$

\noindent and, as one can easily verify, this implies  $dim[\kappa]=time^{-b}$ (alternatively, one can show that  
$dim[\kappa]=time^{-b}$  by simply noting that $k_N = \kappa\delta_N^b$ is dimensionless and $dim[\delta_N^{b}]=time^{b}$).

Now observe that the leading term of any expansion of the PSRV finite-sample bias must have dimension equal to $1/time{^2}$. Based on this observation, we conjecture that the leading term of the expansion in Theorem \ref{th4} under Assumption \ref{Ass4}  is  

   $$\Big(\frac{4 E[\nu(\tau)]^2}{\kappa^2\delta_N^{1+2b}} -\gamma^2E[\nu(\tau)^{2\beta}]\Big)h,$$

\noindent whose dimension is $1/time^2$, as one can easily check by recalling that $dim[\kappa]=time^{-b}$, $\dim[\nu(t)]=1/time$ and $dim[\gamma] =1/time^{-\beta+3/2}$. Accordingly, if one conditions the bias to the natural filtration of $\nu(t)$ up to time $t=\tau$, the  generalized bias-optimal value of $\kappa$, for $b=-1/2$ and $c<1/2$, reads
\begin{equation}\label{optimal_general}
 \kappa^{**} :=2\frac{ \nu(\tau)^{1-\beta}}{\gamma } ,
\end{equation}

\noindent Note that equation (\ref{optimal_general}) can be rewritten in non-parametric form as 
\begin{equation*}
 \kappa^{**} = 2\frac{\nu(\tau)}{\sqrt{\xi(\tau)}},
\end{equation*}

\noindent where $\xi(t):=\gamma^2\nu(t)^{2\beta}$ is the vol-of-vol process. This result, while offering insight into the non-parametric solution to the  problem of the bias-optimal selection of $\kappa$, is problematic in terms of feasibility as it requires the estimation of the spot vol-of-vol $\xi(t)$ at $t=\tau$, a challenging issue which has not been addressed so far in the literature to the best of our knowledge and goes  beyond the scope of this paper.

Our conjecture is based on the origin of the two addenda in the leading term of the bias  (see Theorem \ref{th4}) in the CIR framework. In fact, bearing in mind the the leading term  is
$$
 \displaystyle  \Bigg(  \frac{4E[\nu(\tau)]^2}{k^2\delta_N^{1+2b}} -\gamma^2 E[\nu(\tau)]    \Bigg)h  \,,
$$
we note that the second addendum, i.e., $\gamma^2E[\nu(\tau)]h$, comes from the  expected quadratic variation  of the volatility process. More specifically, it originates from the leading term of the following expansion:
\begin{eqnarray*}
&&E\Big[ \langle\nu,\nu\rangle_{[\tau,\tau+h]} \Big] =   \gamma^2 \,E[\nu(\tau)] h +o(h), \quad h\to 0.
\end{eqnarray*}
Instead, the first addendum, i.e., $\frac{4E[\nu(\tau)]^2}{k^2\delta_N^{1+2b}}$, is due to the drift of the volatility process.

 Thus in the case of the CKLS model, the first addendum remains unchanged since the drift of the process is the same for any $\beta$, while
the second addendum changes according  {to} the  expected quadratic variation  of the volatility process, {which}, for small $h$,   reads  
\begin{eqnarray*}
&& E\Big[ \langle\nu,\nu\rangle_{[\tau,\tau+h]} \Big] =   \gamma^2 \,E[\nu(\tau)^{2\beta}] h +o(h), \quad h\to 0. 
\end{eqnarray*}
  since 
$E\Big[ \langle\nu,\nu\rangle_{[\tau,\tau+h]} \Big] =   \gamma^2 \int_{\tau}^{\tau+h} \,E[\nu(s)^{2\beta}] ds$.

\section{Numerical results}\label{num}
\subsection{Numerical results in the CIR setting}

 As detailed in Section \ref{anres}, in the absence of microstructure noise and assuming $\nu(\tau)$ to be observable and $\gamma$ to be known, the finite-sample bias of the PSRV is optimized, under Assumption \ref{Ass2} and for any $\nu(0)$, by selecting $b=-1/2$, $c<1/2$ and $\kappa=\kappa^{*}:= 2\sqrt{\nu(\tau)}\gamma^{-1}$. In this subsection, we give numerical confirmation of the optimality of  this rule for the selection of $\kappa$  in three progressively more realistic scenarios, where incremental sources of biases are added.

In the first scenario, we simulate log-price paths under Assumption \ref{Ass2} and compute daily PSRV values from noise-free price observations assuming that the CIR parameters are known and the initial volatility value $\nu(\tau)$ is observable. In this scenario, we use two price sampling frequencies, that is, $\delta_N=1$ minute and $\delta_N= 5$ minutes. Results show that the bias generated by the price discrete sampling is relatively small, e.g., less than $5\%$
if $\delta_N = 1$ minute when $\kappa = \kappa^{*}$ (see Table \ref{tab1}).

In the second scenario, we simulate log-price paths under Assumption \ref{Ass3} and compute PSRV values from noisy prices while assuming that the CIR parameters are known and the initial volatility value $\nu(\tau)$ is observable. As the PSRV is not robust to the presence of noise contaminations in the price process, here we only consider the sampling frequency $\delta_N= 5$ minutes, as recommended in the seminal paper {by} \cite{abdh}, where the authors suggest that this sampling frequency reduces the impact of noise on returns while still falling within a high-frequency framework. Indeed, a comparison of the numerical results obtained in these first two scenarios shows that the impact of  the price noise on the PSRV estimates is relatively small at the 5-minute sampling frequency, when $\kappa =\kappa^{*}$ is used.\\
\indent In the third scenario, we still simulate the log-price path  under Assumption \ref{Ass3}, but the value of the initial volatility, $\nu(\tau)$, is now unobservable and the  model parameter $\gamma$ is unknown.  Thus, we compute PSRV values from noisy prices by selecting $\kappa=\hat\kappa^{*}:=  \frac{2\sqrt{\hat \nu (\tau)}}{\hat\gamma}$. Here, $\hat\nu(\tau)$ and $\hat\gamma$ are obtained through the estimation procedure detailed in Appendix B.
A comparison of the results obtained in these different scenarios shows that 
the  PSRV finite-sample bias reduction obtained with the feasible selection  $\kappa=\hat\kappa^{*}$   is very similar to the reduction obtained with the unfeasible selection $\kappa=\kappa^{*}$. For the simulation of each scenario, we use the three realistic sets of parameters from Section \ref{mot}. For each parameter set, we simulate one thousand 1-year trajectories of 1-second observations.  

The noise component $\eta$ in Assumption \ref{Ass3} is simulated as an i.i.d.\ Gaussian process, with noise-to-signal ratio $\zeta$ ranging from 0.5 to 3.5, as in the numerical exercise proposed in \cite{scm}. We define the noise-to-signal ratio  $\zeta$ as in \cite{scm}, i.e., $\zeta := \frac{std(\Delta \eta)}{std(r)}$, where $\Delta \eta$ denotes a generic increment of the i.i.d.\ process $\eta$ under Assumption \ref{Ass3} and $r$ denotes the noise-free log-return at the maximum sampling frequency available, which is equal to 1 second in our numerical exercise. From the simulated prices, we compute daily PSRV values, that is, we set a small time horizon $h$, i.e., $h=1/252$. Recall that the bias-optimal rule for the selection of $\kappa$ is valid when $b=-1/2$ and $c<1/2$. Accordingly, we set $b=-1/2$ and $c=1/4$ in our numerical study. \\ 
\indent Tables \ref{tab1}--\ref{tab3} summarize the results of our numerical exercises and, to make the results of the three parameter sets  comparable,  we report the values of the relative bias. Since we simulate 6-hr days, $N$ is equal to $360$ when $\delta_N=1$ minute and $72$ when $\delta_N= 5 $ minutes.  Note that the overlapping condition $W_N>\Delta_N$ is always satisfied for the values of $\Delta_N$ in Table \ref{tab1}. In particular, the average length of $W_N$ is approximately equal to: 530 minutes for \textit{Set 1}, 410 minutes for \textit{Set 2} and 580 for \textit{Set 3}, when $\delta_N = 1 $ minute; 1200 minutes for \textit{Set 1}, 930 minutes for \textit{Set 2} and 1310 minutes for \textit{Set 3}, when $\delta_N =5$  minutes. These averages are computed over all simulated days and are stable across the three scenarios. Recall that the length of $W_N$ varies by day, as it depends on $\kappa^{*}$, which in turn depends on the volatility value at the beginning of each day, i.e., $\nu(\tau)$ (in scenarios 1 and 2), or its estimate, i.e., $\hat\nu(\tau)$  (in scenario 3). 

 \begin{table}[h!]
  
\begin{tabular}{   c | c | c | c | c | c  | c    } 
 
   noise-to-signal ratio $\zeta$ & $\delta_N$ &  $\Delta_N$ & $\lambda$ & rel. bias  1\textit{(Set 1}) &  rel. bias (\textit{Set 2}) & rel. bias  (\textit{Set 3})\\ 
\hline \hline 
    $\zeta=0$ & 1 min.   &  $\delta_N $ (1 min.)& $2 \cdot 10^{-4}$ &0.003& 0.004& 0.032\\ 
         &    & $2\delta_N$ (2 min.) & $4 \cdot 10^{-4}$ &0.006&   0.006& 0.033 \\ 
              &    & $3\delta_N$ (3 min.)  & $6 \cdot 10^{-4}$   &0.008& 0.009& 0.034\\ 
                   &    &$5\delta_N$ (5 min.) & $1 \cdot 10^{-3}$ &0.011& 0.013& 0.036   \\ 
                       &    &$10\delta_N$ (10 min.) & $1.9 \cdot 10^{-3}$     &0.021& 0.025& 0.041 \\  
                          &    & $15\delta_N$ (15 min.)& $2.9 \cdot 10^{-3}$ &0.031&  0.037& 0.047 \\ 
\hline                      
  $\zeta=0 $   & 5 min.  & $\delta_N$ (5 min.)& $6 \cdot 10^{-4}$      & 0.024 & 0.024& 0.060\\ 
                &         & $2\delta_N$ (10 min.) & $1.3 \cdot 10^{-3}$  & 0.029 & 0.029& 0.061 \\ 
                &        & $3\delta_N$ (15 min.) & $1.9 \cdot 10^{-3}$  & 0.031 & 0.033 & 0.061\\ 
                &        & $6\delta_N$  (30 min.)& $3.8 \cdot 10^{-3}$ & 0.046 & 0.049 & 0.063\\
                                     
\end{tabular}
\caption{  Scenario 1: daily PSRV relative bias with $\kappa=\kappa^{*}$,  $\zeta =0$,  $\gamma$  known and $\nu(\tau)$  observable. Model parameters: $\alpha =0.2$, $\theta =5$, $\gamma =0.5$, $\rho=-0.2$, $\nu(0)=0.2$ (\textit{Set 1}); $\alpha =0.03$, $\theta =10$, $\gamma =0.25$, $\rho=-0.8$, $\nu(0)=0.03$ (\textit{Set 2}); $\alpha =0.2$, $\theta =5$, $\gamma =0.5$, $\rho=-0.2$, $\nu(0) = 0.4$ (\textit{Set 3}). }\label{tab1}
\end{table}

\begin{table}   
\begin{tabular}{   c | c | c | c | c | c  | c   }
  noise-to-signal ratio $\zeta$ & $\delta_N$ & $\Delta_N$ & $\lambda$ & rel. bias  (\textit{Set 1}) &  rel. bias (\textit{Set 2}) & rel. bias  (\textit{Set 3})\\ 
\hline \hline                      
  $\zeta=0.5$   & 5 min. & $\delta_N$ (5 min.)& $6 \cdot 10^{-4}$      & 0.025 & 0.024& 0.062\\ 
                &        & $2\delta_N$ (10 min.) & $1.3 \cdot 10^{-3}$  & 0.030 & 0.029& 0.062 \\ 
                &        & $3\delta_N$ (15 min.) & $1.9 \cdot 10^{-3}$  & 0.032 & 0.036 & 0.064\\ 
                &        & $6\delta_N$  (30 min.) & $3.8 \cdot 10^{-3}$ & 0.047 & 0.052 & 0.065\\
                  \hline                      
  $\zeta=1.5$ & 5 min. & $\delta_N$ (5 min.) & $6 \cdot 10^{-4}$   & 0.039 & 0.037& 0.075\\ 
    &   & $2\delta_N$ (10 min.) & $1.3 \cdot 10^{-3}$              & 0.044 & 0.043& 0.076 \\ 
             &   & $3\delta_N$ (15 min.) & $1.9 \cdot 10^{-3}$     & 0.046 & 0.049 & 0.078\\ 
                  &   & $6\delta_N$ (30 min.) & $3.8 \cdot 10^{-3}$ & 0.061 & 0.065 & 0.079\\
\hline                      
  $\zeta=2.5$ & 5 min. & $\delta_N$ (5 min.) & $6 \cdot 10^{-4}$   & 0.064 & 0.064& 0.102\\ 
    &   & $2\delta_N$ (10 min.) & $1.3 \cdot 10^{-3}$              & 0.069 & 0.070& 0.103 \\ 
             &   & $3\delta_N$ (15 min.) & $1.9 \cdot 10^{-3}$     & 0.075 & 0.075 & 0.105\\ 
                  &   & $6\delta_N$ (30 min.) & $3.8 \cdot 10^{-3}$ & 0.091 & 0.091 & 0.107\\ 
\hline                      
  $\zeta=3.5$ & 5 min. & $\delta_N$ (5 min.)& $6 \cdot 10^{-4}$   & 0.108 & 0.105& 0.143\\ 
    &   & $2\delta_N$ (10 min.) & $1.3 \cdot 10^{-3}$              & 0.113 & 0.111& 0.145 \\ 
             &   & $3\delta_N$ (15 min.) & $1.9 \cdot 10^{-3}$     & 0.115 & 0.117 & 0.146\\ 
                  &   & $6\delta_N$ (30 min.)  & $3.8 \cdot 10^{-3}$ & 0.130 & 0.132 & 0.149 
\end{tabular}
\caption{ Scenario 2: daily PSRV  relative bias with $\kappa=\kappa^{*}$,  $\zeta >0$, $\gamma$ known and $\nu(\tau)$  observable. Model parameters: $\alpha =0.2$, $\theta =5$, $\gamma =0.5$, $\rho=-0.2$, $\nu(0)=0.2$ (\textit{Set 1}); $\alpha =0.03$, $\theta =10$, $\gamma =0.25$, $\rho=-0.8$, $\nu(0)=0.03$ (\textit{Set 2}); $\alpha =0.2$, $\theta =5$, $\gamma =0.5$, $\rho=-0.2$, $\nu(0) = 0.4$ (\textit{Set 3}).}\label{tab2}
 
 \vspace{0.5cm}
  
\begin{tabular}{   c | c | c | c | c | c  | c   } 
 
   noise-to-signal ratio $\zeta$ & $\delta_N$ & $\Delta_N$ & $\lambda$ & rel. bias  (\textit{Set 1}) &  rel. bias (\textit{Set 2}) & rel. bias  (\textit{Set 3})\\ 
\hline \hline 

  $\zeta=0.5$ & 5 min. & $\delta_N$ (5 min.)& $6 \cdot 10^{-4}$  &0.059 &0.011 &0.046 \\ 
    &  & $2\delta_N$   (10 min.) & $1.3 \cdot 10^{-3}$ &0.059 & 0.011 &0.047 \\ 
             &   & $3\delta_N$ (15 min.) &  $1.9 \cdot 10^{-3}$ &0.060  &  0.013&  0.047\\ 
                  &   & $6\delta_N$ (30 min.) & $3.8 \cdot 10^{-3}$  &0.060 &0.017 & 0.047 \\   
\hline   
  $\zeta=1.5$ & 5 min. & $\delta_N$ (5 min.)& $6 \cdot 10^{-4}$  &0.068 &0.022 &0.049 \\ 
    &  & $2\delta_N$ (10 min.) &    $1.3 \cdot 10^{-3}$ &0.068 & 0.023 &0.049 \\ 
             &   & $3\delta_N$ (15 min.) &  $1.9 \cdot 10^{-3}$ &0.069  &  0.024&  0.049\\ 
                  &   & $6\delta_N$ (30 min.)  & $3.8 \cdot 10^{-3}$  &0.070 &0.027 & 0.050 \\   
\hline                       
  $\zeta=2.5$ & 5 min. & $\delta_N$ (5 min.)& $6 \cdot 10^{-4}$  &0.085 &0.047 &0.053 \\ 
    &  & $2\delta_N$ (10 min.) &    $1.3 \cdot 10^{-3}$ &0.088 & 0.049 &0.053 \\ 
             &   & $3\delta_N$ (15 min.) &  $1.9 \cdot 10^{-3}$ &0.088  &  0.049&  0.054\\ 
                  &   & $6\delta_N$ (30 min.) & $3.8 \cdot 10^{-3}$  &0.088 &0.051 & 0.054 \\ 
\hline
  $\zeta=3.5$ & 5 min. & $\delta_N$ (5 min)& $6 \cdot 10^{-4}$  &0.112 &0.083 &0.058 \\ 
    &  & $2\delta_N$ (10 min.)&    $1.3 \cdot 10^{-3}$ &0.115 & 0.083 &0.058 \\ 
             &   & $3\delta_N$ (15 min.)&  $1.9 \cdot 10^{-3}$ &0.117  &  0.084&  0.059\\ 
                  &   & $6\delta_N$ (30 min.) & $3.8 \cdot 10^{-3}$  &0.118 &0.088 & 0.061 \\ 
\end{tabular}
\caption{ Scenario 3: daily PSRV  relative bias with $\kappa=\kappa^{*}$,  $\zeta >0$, $\gamma$  unknown and $\nu(\tau)$  unobservable.
 Model parameters: $\alpha =0.2$, $\theta =5$, $\gamma =0.5$, $\rho=-0.2$, $\nu(0)=0.2$ (\textit{Set 1}); $\alpha =0.03$, $\theta =10$, $\gamma =0.25$, $\rho=-0.8$, $\nu(0)=0.03$ (\textit{Set 2}); $\alpha =0.2$, $\theta =5$, $\gamma =0.5$, $\rho=-0.2$, $\nu(0) = 0.4$ (\textit{Set 3}).}\label{tab3}
\end{table}

Table \ref{tab1} shows that for $\delta_N= 1$ minute and  $\Delta_N \le 3$ minutes, the bias is almost negligible (i.e., less than $1\%$) when $\nu(0) =\alpha$, while it is slightly larger but still acceptable (i.e., between $3\%$ and $4\%$) when $\nu(0)=2\alpha$. This is in line with equation \ref{truebiasoverl} in Lemma \ref{explbias}, where it is evident that the source of bias $B_N$ is eliminated when $\nu(0) =\alpha$, which implies $E[\nu(\tau)]=\alpha$. With a price sampling frequency of five minutes, the bias is still acceptable, around $6\%$ at worst. Additionally, Table \ref{tab2} shows that in  the presence of noise, price sampling at five-minute intervals to avoid microstructure frictions represents an acceptable compromise, as the bias is less than $15\%$ even in the presence of very intense microstructure effects. Finally, Table \ref{tab3} shows that the statistical error related to the estimation of $\gamma$ and $\nu(\tau)$ could actually partially compensate for the bias due to the presence of noise, especially when the common assumption $\nu(0)=\alpha$ is violated.

Finally, an important remark is in order. The three realistic parameter sets that we have used  all imply the presence of leverage effects, as each includes a negative $\rho$. To meet the simplifying no-leverage assumption under which the results in Section \ref{anres} are derived, for each scenario  we have also performed additional simulations under the hypothesis of the independence between the Brownian motion driving the price and the volatility, keeping the same values of the variance parameters $\alpha, \theta, \gamma$, and $\nu(0)$ that characterize the scenario.  However, the numerical results obtained in the absence of leverage are basically indistinguishable from those illustrated in Tables \ref{tab1}--\ref{tab3}, thereby suggesting that the leverage is a negligible source of bias. Such additional numerical results are not reported here for brevity, but are available from authors.

 \begin{table} 
\begin{tabular}{   c | c | c | c | c | c  | c   } 
   Model & $\delta_N$ & $\Delta_N$ & $\lambda$ & rel. bias  (\textit{Set 1}) &  rel. bias (\textit{Set 2}) & rel. bias  (\textit{Set 3})\\ 
\hline \hline 
    $\beta=\frac{1}{2}$ & 1 min. &    $\delta_N$ (1 min.)& $2 \cdot 10^{-4}$ &0.014& 0.012& 0.027\\ 
         &   & $2\delta_N$ (2 min.)&  $4 \cdot 10^{-4}$ &0.017&   0.015& 0.029 \\ 
              &   & $3\delta_N$ (3 min.)  & $6 \cdot 10^{-4}$   &0.020 & 0.016& 0.029\\ 
                   &   & $5\delta_N$ (5 min.) & $1 \cdot 10^{-3}$ &0.024 & 0.022& 0.031   \\ 
                       &   &$10\delta_N$ (10 min.) & $1.9 \cdot 10^{-3}$     & 0.034& 0.033& 0.036 \\  
                          &   & $15\delta_N$ (15 min.) & $2.9 \cdot 10^{-3}$ &0.042&  0.044& 0.039 \\ 
                          
\end{tabular}
\caption{$\beta=1/2$: daily PSRV finite-sample relative bias with $\kappa=\kappa^{**}$,  $\zeta =0$, $\gamma$  known and $\nu(\tau)$  observable. Model parameters: $\alpha =0.2$, $\theta =5$, $\gamma =0.5$, $\rho=-0.2$, $\nu(0)=0.2$ (\textit{Set 1}); $\alpha =0.03$, $\theta =10$, $\gamma =0.25$, $\rho=-0.8$, $\nu(0)=0.03$ (\textit{Set 2}); $\alpha =0.2$, $\theta =5$, $\gamma =0.5$, $\rho=-0.2$, $\nu(0) = 0.4$ (\textit{Set 3}). The price drift, $\mu$, is always equal to $0.05$. }\label{tab4}

\end{table}

\begin{table}
   
\begin{tabular}{   c | c | c | c | c | c  | c   } 
 
  Model & $\delta_N$ & $\Delta_N$ & $\lambda$ & rel. bias  (\textit{Set 1}) &  rel. bias (\textit{Set 2}) & rel. bias  (\textit{Set 3})\\ 
\hline \hline 
    $\beta=1$ & 1 min. &    $\delta_N$ (1 min.)& $2 \cdot 10^{-4}$ &0.003& 0.002& 0.011\\ 
         &   & $2\delta_N$ (2 min.)&  $4 \cdot 10^{-4}$ &0.004&   0.002& 0.012 \\ 
              &   & $3\delta_N$ (3 min.)   & $6 \cdot 10^{-4}$   &0.005& 0.002& 0.014\\ 
                   &   & $5\delta_N$ (5 min.) & $1 \cdot 10^{-3}$ &0.006& 0.003& 0.015   \\ 
                       &   &$10\delta_N$ (10 min.) & $1.9 \cdot 10^{-3}$     &0.008& 0.005& 0.017 \\  
                          &   & $15\delta_N$ (15 min.) & $2.9 \cdot 10^{-3}$ &0.012&  0.006& 0.021 \\ 
                          
\end{tabular}
\caption{ $\beta=1$: daily PSRV finite-sample relative bias with $\kappa=\kappa^{**}$, $\gamma$ known and $\nu(\tau)$  observable. Model parameters: $\alpha =0.2$, $\theta =5$, $\gamma =0.5$, $\rho=-0.2$, $\nu(0)=0.2$ (\textit{Set 1}); $\alpha =0.03$, $\theta =10$, $\gamma =0.25$, $\rho=-0.8$, $\nu(0)=0.03$ (\textit{Set 2}); $\alpha =0.2$, $\theta =5$, $\gamma =0.5$, $\rho=-0.2$, $\nu(0) = 0.4$ (\textit{Set 3}). The price drift, $\mu$, is always equal to $0.05$.  }\label{tab5}
\end{table}

\subsection{Numerical results in the more general CKLS setting}

\normalsize
We conclude this section by testing the efficacy of the generalized, conjecture-based, criterion for the bias-optimal selection of $\kappa$ under Assumption \ref{Ass4}, i.e., under the assumption that the volatility evolves as a CKLS model. In this case, the feasible version of the bias-optimal rule to select $\kappa$ is given by $\hat\kappa^{**}=\frac{2 \hat\nu(\tau)^{1-\beta}}{\hat\gamma }$, for $b=-1/2$, $c<1/2$. 

 To test the efficacy of this criterion, we repeat the numerical exercise previously performed in scenario 1 under Assumption \ref{Ass2}, considering three different values of $\beta$: $\beta=1/2$, corresponding to the model by \cite{heston}, which differs from the model of Assumption \ref{Ass2} only in the presence of a price drift; $\beta=1$, corresponding to the continuous-time GARCH model by \cite{nelson}; and  $\beta=3/2$, corresponding to the 3/2 model by \cite{32}. For all parameter sets, $\mu$ is set equal to 0.05. Tables 4,5 and 6 show that our general criterion for the bias-optimal selection of $\kappa$ under Assumption \ref{Ass4} is effective, as it gives satisfactory results in terms of relative bias. Note that the case $\beta=1/2$ is of interest only in that it confirms that the criterion for the bias-optimal selection of $\kappa$ derived analytically under Assumption \ref{Ass2}, i.e., $\kappa = \kappa^{*}$, is also effective in the presence of a price drift.

  \vspace{0.25cm}  
\begin{table}  
\begin{tabular}{   c | c | c | c | c | c  | c   } 
 
   Model & $\delta_N$ & $\Delta_N$ & $\lambda$ & rel. bias  (\textit{Set 1}) &  rel. bias (\textit{Set 2}) & rel. bias  (\textit{Set 3})\\ 
\hline \hline 
    $\beta=\frac{3}{2}$ & 1 min. &  $\delta_N$ (1 min.)& $2 \cdot 10^{-4}$ &0.004& 0.001& 0.029\\ 
         &   & $2\delta_N$ (2 min.) &  $4 \cdot 10^{-4}$ &0.004&   0.002& 0.031 \\ 
              &   & $3\delta_N$ (3 min.) & $6 \cdot 10^{-4}$   &0.005& 0.003& 0.031\\ 
                   &   & $5\delta_N$ (5 min.) & $1 \cdot 10^{-3}$ &0.006& 0.006& 0.037  \\ 
                       &   &$10\delta_N$ (10 min.)& $1.9 \cdot 10^{-3}$     &0.006& 0.007& 0.038 \\  
                          &   & $15\delta_N$ (15 min.) & $2.9 \cdot 10^{-3}$ &0.008&  0.009& 0.041 \\ 
                          
\end{tabular}
\caption{$\beta=3/2$: daily PSRV finite-sample relative bias with $\kappa=\kappa^{**}$, $\gamma$  known and $\nu(\tau)$ observable. Model parameters: $\alpha =0.2$, $\theta =5$, $\gamma =0.5$, $\rho=-0.2$, $\nu(0)=0.2$ (\textit{Set 1}); $\alpha =0.03$, $\theta =10$, $\gamma =0.25$, $\rho=-0.8$, $\nu(0)=0.03$ (\textit{Set 2}); $\alpha =0.2$, $\theta =5$, $\gamma =0.5$, $\rho=-0.2$, $\nu(0) = 0.4$ (\textit{Set 3}). The price drift, $\mu$, is always equal to $0.05$. }\label{tab6}
\end{table}

\section{Empirical study}\label{emp}
We conclude the paper with an empirical analysis, where we apply the bias-optimal criterion for selecting $\kappa$ in equation (\ref{optimal_general})  to compute daily PSRV estimates. The dataset is composed of  two 1-year samples of  S\&P 500 1-minute prices relative to the years 2016 and 2017, respectively. 
The two samples are analyzed separately since the volatility of these two time series behaves very differently. In fact, the year 2016 is characterized by volatility spikes (due, e.g., to  uncertainty pertaining to the so-called Brexit in the month of June or the U.S. presidential election in the month of November), while the year 2017 is  characterized by low volatility, as one can see in  Figure \ref{data}. Analyzing  the two series separately allows for validation of the feasible rule for the selection of $\kappa$ in two very different scenarios.

\begin{figure}[H]
\includegraphics[width=\linewidth]{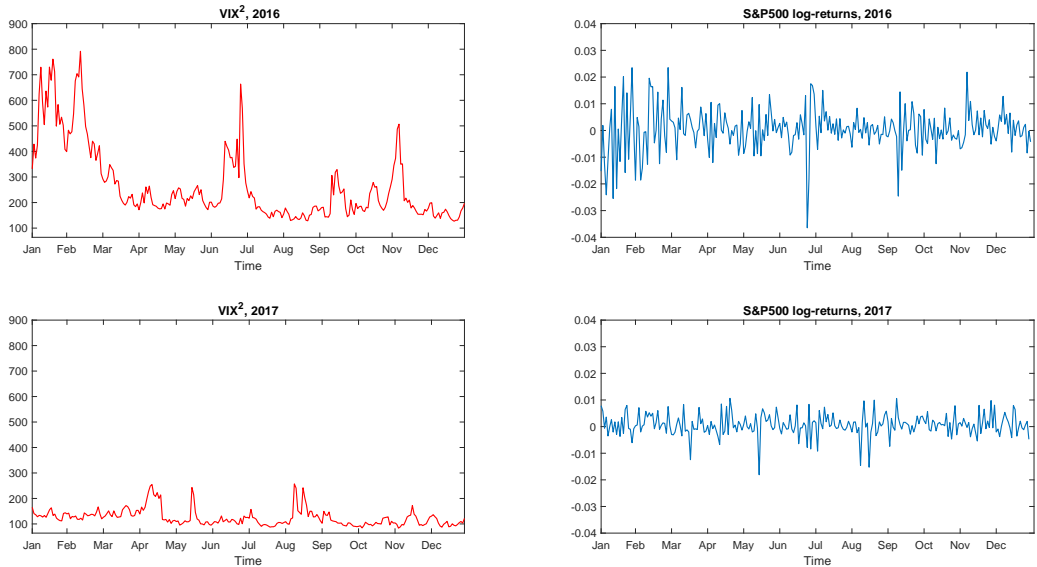}
\caption{Daily VIX$^2$ values (left) and daily S\&P 500 log-returns (right) in the years 2016 and 2017.} \label{data}
\end{figure}

We proceed as follows. First, through the method detailed in Appendix B, we obtain non-parametric estimates of the process $\nu$ at the beginning of each day and estimates of $\gamma$ under Assumption \ref{Ass4}, for the three different values of $\beta$ considered in the numerical exercise of Section \ref{num}. The results of the estimation of $\gamma$ are shown in Table \ref{tab7}\footnote{The estimates of the process $\nu$ at the beginning of each day are not reported for the sake of brevity. See Chapter 4 in \cite{mrs} for a detailed study which demonstrates the finite-sample accuracy of the Fourier estimator of the spot volatility suggested in Appendix B.}.

\begin{table}[H]
 \begin{center}
 
\begin{tabular}{   c | c | c |      c    } 
 
   Model & Sample year & $\hat \gamma $   & $R^2$\\ 
\hline \hline 
  $\beta=\frac{1}{2}$ & 2016 & 0.7127     &   0.1383  \\ 
     & 2017 & 0.4250  &      0.1736  \\ 
     \hline
 $\beta=1$         & 2016   & 6.1682  &    0.0725  \\ 
          &  2017 & 5.9973&    0.0901  \\
 \hline 
  $\beta=\frac{3}{2}$   &   2016   &  72.0358   &0.1141   \\ 
    &  2017    &  65.1084 &     0.0866  \\                          
\end{tabular}
\caption{ Results of the estimation of $\gamma$  under Assumption \ref{Ass4} for different values of $\beta$. }\label{tab7}
\end{center}
\end{table}

Then, based on $R^2$ values, we assume the Heston model ($\beta=1/2$) as the data generating process for both samples. Consequently, we select $b=-1/2$, $c=1/4$, $\kappa = 2 \hat \gamma^{-1}\sqrt{\hat\nu(\tau)}$ and compute daily PSRV values from  empirical prices sampled at the frequency $\delta_N=5$ minutes. The resulting selection of $W_N$ is approximately equal, on average, to $450$ minutes in 2016 and $275$ minutes in 2018. Note that the selection $\delta_N= 5$ minutes is justified by the fact that we assume the impact of microstructure contaminations to be negligible at that sampling frequency, based on the application of the Hausman test by \cite{asx} for the presence of noise, which tells that the impact of noise at the 5-minute frequency is negligible in our samples, confirming a well-known stylized fact (see \cite{abdh}).

Additionally, we have performed the jump-detection test by \cite{cpr}. Precisely, we have performed this test on 5-minute returns, as it is not robust to the presence of noise.  Based on the results of the jump test, we have computed daily PSRV values from 5-minute prices according to the following procedure for the removal of days with jumps\footnote{Note that the analytical results in Section \ref{anres} are derived under the assumption of absence of jumps in the  price and volatility. The literature on non-parametric jump tests provides large and robust empirical evidence, mainly based on US markets, that volatility jumps  are accompanied by price jumps (see, e.g., \cite{JT,BRcoJ,BW}). Thus removing days with price jumps from the estimation  basically also takes care of jumps in the volatility. In the absence of jumps, a model in the CKLS class for the volatility could provide a reasonable trade-off between accuracy in reproducing empirical features of prices and parsimony in terms of parameters to be estimated, as pointed out, e.g., in \cite{Christ2010} and \cite{Goard}.}.  Assume that  time is measured  in days and that we are interested in computing the daily PSRV on  $[\tau, \tau +1]$. Further, without loss of generality, assume that the bias-optimal value of $W_N$ is equal to 2.5 days for $\delta_N$ equal to 5 minutes. If jumps have occurred within the interval $[\tau-1,\tau+1]$, i.e., on the day of interest or the day before,  we do not compute the PSRV on $[\tau,\tau+1]$; if jumps have occurred in $[\tau-2,\tau-1)$, but not in $[\tau-1,\tau+1],$ at any instant $\tau+i\Delta_N$, $i=0,1,\ldots,\lfloor \Delta_N^{-1} \rfloor$,  we select $W_N=1+i\Delta_N \le2$ {to pre-estimate} the spot volatility; instead, if   jumps have occurred in $[\tau-3,\tau-2)$, but not in $[\tau-2,\tau+1]$,  at any instant $\tau+i\Delta_N$, we  select $W_N=min(2+i\Delta_N,2.5)$; finally, if no jump has occurred within the period $[\tau-3,\tau+1]$,  at any instant $\tau+i\Delta_N$, we select $W_N=2.5.$ Overall, based on this procedure, the days for which we do not compute the PSRV amount to $12.25\%$ of the sample in 2016 and  $8.30\%$ of the sample in 2017.

Figures \ref{2016} and \ref{2017} show the daily PSRV values obtained for four different values of $\lambda$ corresponding to a spot volatility estimation frequency $\Delta_N$ equal to 5, 10, 15, and 30 minutes, respectively.

\begin{figure}[hpt]
\includegraphics[width=\linewidth]{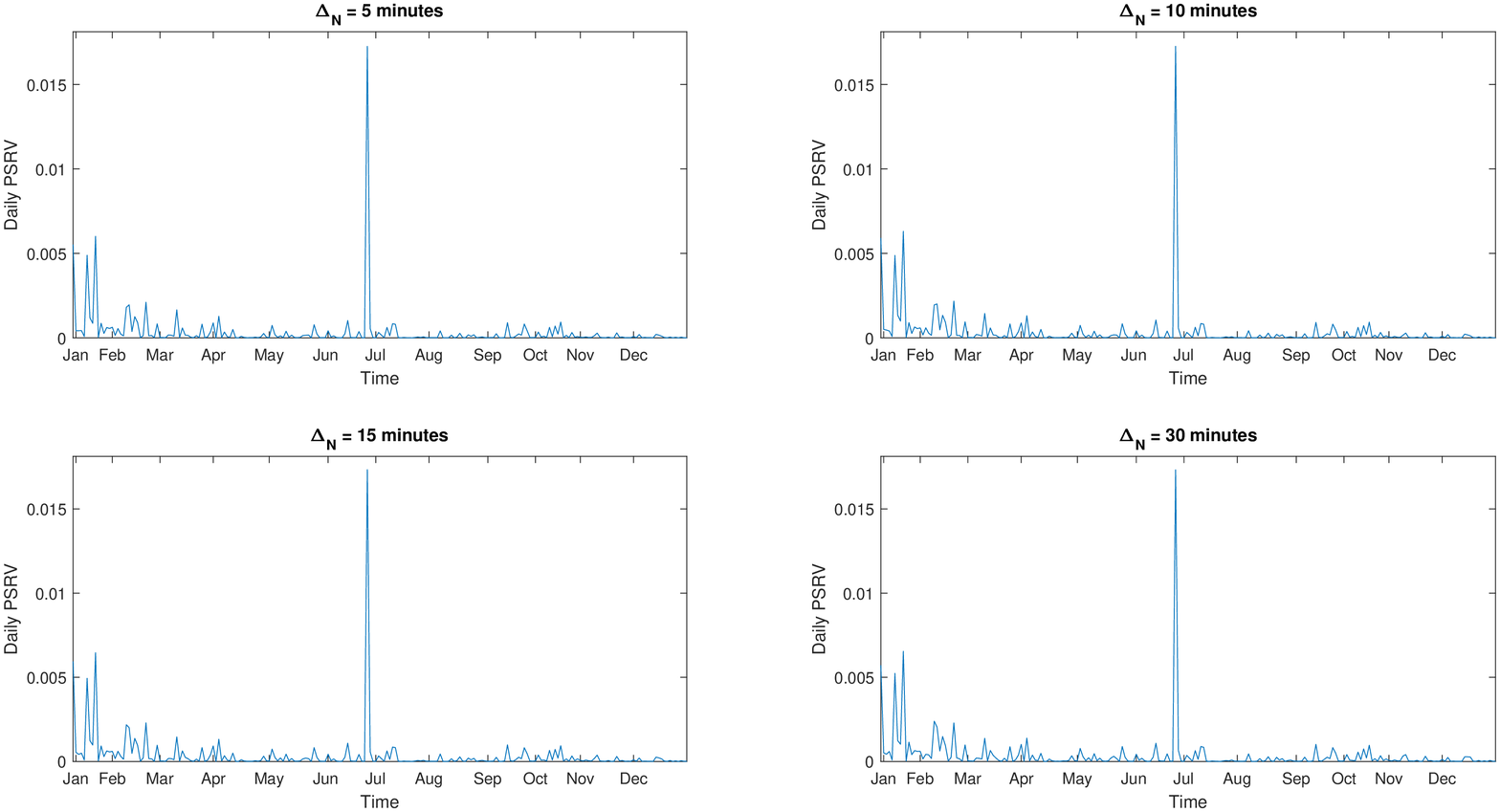}
\caption{Daily PSRV values in the year 2016.} \label{2016}
\end{figure}

\begin{figure}[hpt]
\includegraphics[width=\linewidth]{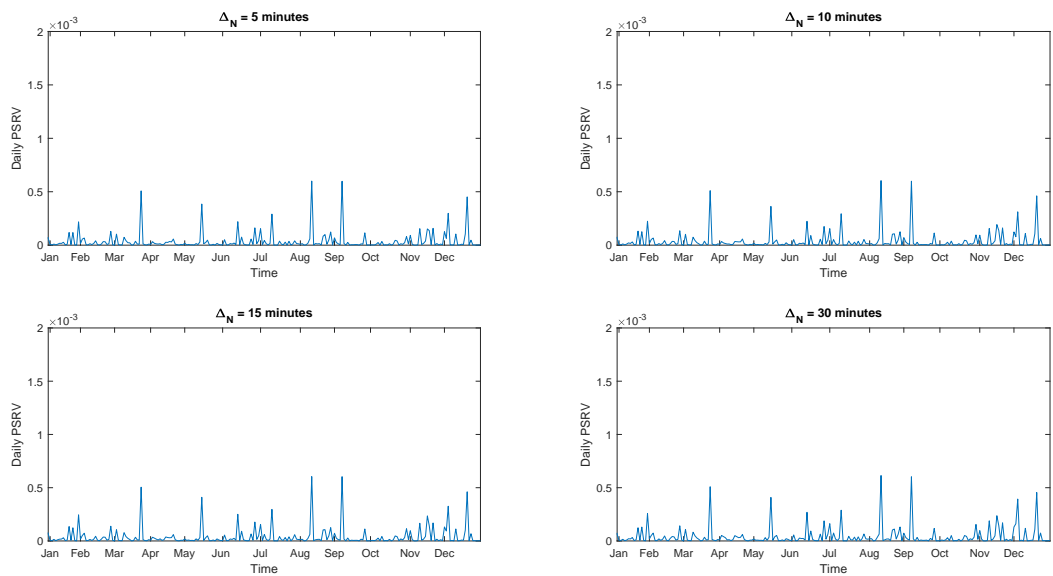}
\caption{Daily PSRV values in the year 2017.} \label{2017}
\end{figure}

Comparing the dynamics of the VIX$^2$ index in Figure \ref{data} with those of the PSRV, one notices that when the VIX$^2$
spikes, the vol-of-vol also spikes (see, e.g., the behavior of the plots at the end of June 2016) and, viceversa, when the VIX$^2$ is low and stable (e.g., in 2017) the vol-of-vol is also low and stable. This evidence corroborates the goodness of our vol-of-vol estimates. Finally, note that 
for either of the two samples, the plots for different values of $\Delta_N$ are basically indistinguishable. With respect to the bias-optimal selection of $\lambda$ (i.e., $\Delta_N$), this evidence confirms what emerges from the analytic study in Section \ref{anres}: the impact of the selection of $\lambda$ (i.e., $\Delta_N$) on PSRV values is marginal, if not negligible.

\section{Conclusions}\label{concl}

The pre-estimated spot-variance based realized variance (PSRV) by \cite{bnv}, the simplest and most natural consistent estimator of the integrated vol-of-vol, is typically affected by a substantial finite-sample bias.  The main contribution of this paper is  to  show, analytically, that  local-window overlapping in finite samples effectively reduces this bias. This result confirms the findings of \cite{scm}, based on simulations.    

The paper is written in the spirit of \cite{afl}. In \cite{afl}, a parametric data-generating process, namely the Heston model, is used to obtain a fully explicit bias expression for the price-volatility correlation, the most natural leverage estimator, which is very biased at high frequencies. Based on the full explicit knowledge of the bias, the authors are able to isolate the sources of bias that affect the simple leverage estimator and derive a feasible strategy to correct for them. In this paper we follow a similar approach. Assuming that the volatility is a CIR process, we obtain the full explicit expression of the PSRV finite-sample bias. Crucially, we show that this expression differs in the overlapping case and the no-overlapping case and, most importantly, that a feasible bias-correction strategy for finite samples can be derived only in the overlapping case. 

Further, using dimensional analysis, we generalize the feasible bias-correction strategy to hold under the assumption that the volatility process belongs to the more general  CKLS class, which encompasses a number of widely-used parametric models. Numerical results corroborate the validity of the generalized rule in that nearly unbiased vol-of-vol estimates are obtained for two  other models  in the CKLS class, namely, the continuous-time GARCH model and the 3/2 model.

In the paper, the impact of microstructure noise on the PSRV bias is also investigated. First, we derive  the exact analytic expression of the extra bias due to noise, which differs in the overlapping and no-overlapping cases. This extra bias can not be consistently estimated over a fixed time horizon and then subtracted, as it depends, other than on some moments of the noise process, on the drift of the volatility.  As a solution, we propose to apply the feasible rule for the bias-optimal selection of the local-window parameter on sparsely-sampled prices, following \cite{abdh}. Numerical evidence of the efficacy of this solution is provided.

Finally, as a byproduct of this analysis, we quantify, for both the PSRV and the locally averaged realized variance,  the  bias reduction  ensuing from  the assumption that the initial value of the volatility  is equal to its long-term mean,  which is very common in simulation studies found in the literature.

\section*{Declarations}

\textit{Funding:} Not applicable. 

\noindent \textit{Conflicts of interest:} The authors have no conflicts of interests to declare. 
 
\noindent\textit{Availability of data and material:} Not applicable. 
 
\noindent\textit{Code availability:} See supplementary file.

\bibliographystyle{apalike}
\bibliography{ref}

 \begin{appendices}
 \vspace{1cm}
\section{Proofs}

\subsection*{Lemma \ref{explbias}}
\begin{proof} 
 
 From Definition \ref{psrv} we have
 $$ PSRV_{[\tau,\tau+h],N }   := \displaystyle \sum_{i=1}^{\lfloor h / \Delta_N \rfloor}  \Big[ \hat{\nu} (\tau+i\Delta_N)-\hat{\nu} (\tau+ i\Delta_N-\Delta_N ) \Big]^2, $$
where, for $s$ taking values on the time grid of mesh-size $\delta_N$:

\begin{itemize}
\item[-]  $\hat{\nu} (s):=RV (s, k_N\delta_N)(k_N\delta_N)^{-1},$  
\item[-]  $RV (s, k_N\delta_N) := \displaystyle \sum_{j=1}^{k_N}  \Delta p^2 (s +  j\delta_N-k_N\delta_N, \delta_N ),$
\item[-] $\Delta p (s,\delta_N  ):= p(s  )-p(s -\delta_N  ).$ 

\end{itemize}

\noindent Note that $E[PSRV_{[\tau,\tau+h],N }]$ can be rewritten as

\begin{equation}\label{A1}
E[PSRV_{[\tau,\tau+h],N }]  = (k_N\delta_N)^{-2}\sum_{i=1}^{\lfloor h / \Delta_N \rfloor}  E[ RV^2 (\tau+i\Delta_N, k_N\delta_N)] + E[RV^2 (\tau+i\Delta_N \\-\Delta_N, k_N\delta_N)]$$
$$ -2E[RV(\tau+i\Delta_N, k_N\delta_N)RV(\tau+ i\Delta_N- \Delta_N, k_N\delta_N) ].
\end{equation}

Therefore, under Assumption \ref{Ass2}, the  explicit formula for $E[PSRV_{[\tau,\tau+h],N }] $  can be obtained by deriving the analytic expression for $E[RV^2 (\tau+i\Delta_N, k_N\delta_N)]$, $E[RV^2 (\tau+ i\Delta_N -\Delta_N, k_N\delta_N)]$ and
$E[RV(\tau +i\Delta_N, k_N\delta_N)RV(\tau +i\Delta_N- \Delta_N, k_N\delta_N) ] $. Note that the expression of the last term differs in the no-overlapping case $W_N \le \Delta_N$ and the overlapping case $W_N >\Delta_N$. \\
We derive the exact expression of these terms separately as follows.

\vspace{0.5cm}
 
 \textbf{I)  Analytic expression of} $\boldsymbol{E[RV^2 (\tau+i\Delta_N, k_N\delta_N)]}$ 
 \vspace{0.25cm}
 
To {simplify} the notation, let $a_{i,u,N}$ {denote} the quantity  $a_{i,u,N} =\tau+i\Delta_N + (u-k_N)\delta_N $. Also, let $(\mathcal{F}^{v}_{s})_{s\ge0}$ be the natural filtration associated with the process $\nu$. We have

$$\displaystyle{ E[RV^2 (\tau+i\Delta_N, k_N\delta_N)] =  \sum_{j=1}^{k_N} E[ \Delta p ^4(\tau+i\Delta_N + (j-k_N)\delta_N)]}+ 2\displaystyle{ \sum_{j=2}^{k_N} E[ \Delta p ^2(\tau+i\Delta_N + (j-k_N)\delta_N) \displaystyle\sum_{h=1}^{j-1}  \Delta p^2(\tau+i\Delta_N + (h-k_N)\delta_N)]}$$

$$=\displaystyle \sum_{j=1}^{k_N} E\Big[ \Big(\int_{a_{i,j-1,N}}^{a_{i,j,N}} \sqrt{v(s)} dW(s)\Big)^4 \Big]+  2\displaystyle \sum_{j=2}^{k_N}\displaystyle\sum_{h=1}^{j-1}  E \Big[ \Big(\int_{a_{i,j-1,N}}^{a_{i,j,N}} \sqrt{v(s)} dW(s) \Big)^2 \times \displaystyle \Big(\int_{a_{i,h-1,N}}^{a_{i,h,N}} \sqrt{v(s)} dW(s) \Big)^2\Big],$$
 where:
 
 \begin{itemize}
\item[$\bullet$] $\displaystyle\int_{a_{i,j-1,N}}^{a_{i,j,N}} \sqrt{v(s)} dW(s) | \mathcal{F}^{v}_{a_{i,j,N}} \sim\mathcal{N}\Big(0,\displaystyle \int_{a_{i,j-1,N}}^{a_{i,j,N}} v(s)ds \Big),$ since $Z$ is assumed to be independent of $W$ (see Sections 3.1.3 and 3.1.4  of \cite{aj}); this in turn implies \\

$  E\Big[ \Big(\displaystyle\int_{a_{i,j-1,N}}^{a_{i,j,N}} \sqrt{v(s)} dW(s)\Big)^4 \Big]= E\Big[ E\Big[\Big( \displaystyle\int_{a_{i,j-1,N}}^{a_{i,j,N}} \sqrt{v(s)} dW(s) \Big)^4| \mathcal{F}^{v}_{a_{i,j,N}} \Big]\Big] = 3E\Big[ \Big( \displaystyle\int_{a_{i,j-1,N}}^{a_{i,j,N}} v(s) ds \Big)^2 \Big];  $

\item[$\bullet$] for $h < j $ and $s<r$, \\
 
$E \Big[\displaystyle \Big(\int_{a_{i,j-1,N}}^{a_{i,j,N}} \sqrt{v(s)} dW(s) \Big)^2 \Big(\int_{a_{i,h-1,N}}^{a_{i,h,N}} \sqrt{v(s)} dW(s) \Big)^2\Big]$\\
$=E \Big[E\Big[\displaystyle \Big(\int_{a_{i,j-1,N}}^{a_{i,j,N}} \sqrt{v(s)} dW(s) \Big)^2 \Big(\int_{a_{i,h-1,N}}^{a_{i,h,N}} \sqrt{v(s)} dW(s) \Big)^2| \mathcal{F}^{v}_{a_{i,j,N}} \Big]\Big]$\\
$ = \displaystyle \int_{a_{i,j-1,N}}^{a_{i,j,N}} \int_{a_{i,h-1,N}}^{a_{i,h,N}}E [v(s)E [ v(r)| \mathcal{F}_s^\nu]]ds dr  .$

 \end{itemize}

 \vspace{0.25cm}
Under Assumption \ref{Ass2} (see Appendix A in \cite{bz}), we also have:

\begin{itemize}
\item[$\bullet$] $ E\Big[ \Big( \displaystyle\int_{a_{i,j-1,N}}^{a_{i,j,N}} v(s) ds \Big)^2 \Big] =   \frac{1}{\theta^{2}}  (1-e^{-\theta \delta_{N}} )^{2} \Big\{e^{-2 \theta i \Delta_{N}-2 \theta j \delta_{N}+2 \theta (1+k_{N} ) \delta_{N} } E[\nu(\tau) ]^{2} $ \\ 
$ +\Big[\frac{\gamma^{2}}{\theta} (e^{-\theta i \Delta_{N}-\theta j \delta_{N}+\theta (k_{N}+1 ) \delta_{N}}-e^{-2 \theta i \Delta_{N}-2 \theta j \delta_{N}+2 \theta (k_{N}+1 ) \delta_{N}} )$\\
$+  2 \alpha e^{-\theta i \Delta_{N}-\theta j \delta_{N}+\theta (k_{N}+1 ) \delta_{N}}  ( 1-e^{ -\theta i \Delta_{N}-\theta j \delta_{N}+\theta ( k_{N}+1 ) \delta_{N}  } ) \Big] E[\nu(\tau)]$\\ 
$ + \Big(\frac{\gamma^{2} \alpha}{2 \theta}+\alpha^{2}\Big ) (1+e^{-2 \theta i \Delta_{N}-2 \theta j \delta_{N}+2 \theta (k_{N}+1) \delta_{N}}  -2 e^{-\theta_{i} \Delta_{N}-\theta j \delta_{N}+\theta (k_{N}+1 ) \delta_{N}} )\Big\} $\\
$+  \frac{\gamma^{2}}{\theta^{2}} \Big(\frac{1}{\theta}-2\delta_{N} e^{-\theta \delta_{N}}  -\frac{1}{\theta} e^{-2 \theta \delta_{N}} \Big)+2 \frac{1}{\theta} (1-e^{-\theta \delta_{N}} ) \Big[\alpha \delta_{N}-\frac{\alpha}{\theta} (1-e^{-\theta \delta_{N}} ) \Big] \\ \times \Big[ e^{-\theta i \Delta_{N}-\theta j \delta_{N}+\theta (k_{N}+1 ) \delta_{N}} E[\nu(\tau)]+\alpha (1-e^{ -\theta i \Delta_{N}-\theta j \delta_N+\theta (k_{N}+1 ) \delta_{N}  })  \Big ]\\ +\frac{\gamma^{2}}{\theta^{2}} \Big[\alpha \delta_{N} (1+2 e^{-\theta \delta_{N}} )+\frac{\alpha}{2 \theta} (e^{-2 \theta \delta_{N}}+4 e^{-\theta \delta_{N}}-5 ) \Big]+\alpha^{2} \delta_{N}^{2}+\frac{\alpha^{2}}{\theta^{2}} (1-e^{-\theta \delta_{N}} )^{2}   -2 \frac{\alpha^{2}}{\theta} \delta_{N} (1-e^{-\theta \delta_{N}} )     ;$

 \vspace{0.25cm}

\item[$\bullet$]for $h < j $ and $s<r$, \\
 
$ \displaystyle \int_{a_{i,j-1,N}}^{a_{i,j,N}} \int_{a_{i,h-1,N}}^{a_{i,h,N}}E [v(s)E [ v(r)| \mathcal{F}_s^\nu]]ds dr   $\\
$=\Big[\Big( E[\nu(\tau)]-\alpha\Big)^{2} + \frac{\gamma^{2}}{\theta} \Big(\frac{\alpha}{2}-E[\nu(\tau)]\Big)\Big] \frac{1}{\theta^{2}}\cdot e^{-2 \theta i \Delta_{N} -\theta j   \delta_{N} -\theta h \delta_{N}+2 \theta {k_{N}} \delta_{N}} (1-e^{\theta {\delta_N  }} )^{2} $
 \\
$-\frac{\gamma^{2} \alpha}{2\theta^{3}} e^{-\theta j \delta_{N}+\theta h \delta_{N}} (2-e^{-\theta \delta_{N}}-e^{\theta \delta_{N}})+\alpha^{2} \delta_{N}^{2}   -\frac{\gamma^2}{\alpha \theta}(E[\nu(\tau)]-\alpha)  \delta_{N} (1-e^{\theta \delta_{N}} ) 
 e^{-\theta i \Delta_{N}-\theta j \delta_{N}+\theta {k_N} \delta_{N}} $\\
 $-\frac{\alpha}{\theta}(E[\nu(\tau)]-\alpha)  \delta_{N} (1-e^{\theta \delta_{N}} ) e^{-\theta i\Delta_N-\theta h \delta_{N}+ \theta k_N \delta_{N}}$. 
  
\end{itemize}

Finally, putting everything together, we obtain the following expression for $E[RV^2 (\tau+i\Delta_N, k_N\delta_N)]$: 
 \vspace{0.25cm}

\noindent $E[RV^2 (\tau+i\Delta_N, k_N\delta_N)]$\\
$=\displaystyle  (1-e^{-2 \theta k_{N} \delta_{N}} ) (1-e^{-2 \theta \delta_{N}} )^{-1} e^{-2 i \Delta_{N}+2 \theta k_{N}  \delta_{N}} (1-e^{-\theta \delta_{N}} )^{2} \frac{3}{\theta^{2}} \Big[ (E[\nu(\tau)]-\alpha)^{2}+\frac{\gamma^{2}}{\theta} \Big(\frac{\alpha}{2}-E[\nu(\tau)] \Big) \Big]$\\
$+  (1-e^{-\theta  k_N \delta_{N}} ) (1-e^{-\theta \delta_{N}} )^{-1} e^{-i \Delta_{N}+\theta k_{N}  \delta_{N}} \Big\{ \frac{\gamma^{2}}{\theta}(E[\nu(\tau)]-\alpha) \frac{3}{\theta^{2}} (1-e^{-\theta \delta_{N}} )^{2}$\\
$+\frac{\gamma^{2}}{\theta} \Big(\frac{1}{\theta}-2 e^{-\theta \delta_{N}} \delta_{N}-\frac{1}{\theta} e^{-2 \theta \delta_{N}} \Big) \frac{3}{\theta}(E[\nu(\tau)]-\alpha)  + \Big[6 \frac{\alpha}{\theta} \delta_{N} (1-e^{- \theta \delta_{N}} )  \Big](E[\nu(\tau)]-\alpha) \Big \}$\\
$ +\frac{\gamma^{2}}{\theta} k_{N} \Big[ \frac{3\alpha}{2\theta^2}  (1-e^{-\theta \delta_{N}} )^{2}  +3 \frac{\alpha}{\theta} \Big(\frac{1}{\theta}-2 e^{-\theta \delta_{N}} \delta_{N}-\frac{1}{\theta} e^{-2 \theta \delta_{N}}  \Big)+ 3\frac{\alpha}{\theta}  \delta_{N} (1+2 e^{-\theta \delta_{N}})$\\
$+  \frac{3\alpha}{2\theta^{2}} (e^{-2 \theta \delta_{N}}+4 e^{-\theta \delta_{N}}-5 )  \Big]+3 \alpha^{2} \delta_{N}^{2} k_N$\\
$+2  \Big[\frac{\gamma^{2}}{\theta}  \Big(\frac{\alpha}{2}-E[\nu(\tau)]  \Big)+(E[\nu(\tau)]-\alpha)^{2}  \Big] \frac{1}{\theta^{2}}e^{2 \theta k_{N} \delta_{N} -2 \theta i \Delta_{N}-  \theta \delta_{N}}  $\\
 $\times\displaystyle (1-e^{-\theta (k_{N}-1 ) \delta_{N}}+e^{-\theta (k_{N}+i ) \delta_{N}}-e^{-\theta \delta_{N}}+e^{-\theta (2 k_{N}-1 ) \delta_{N}}  -e^{-2 \theta k_{N} \delta_{N}} )    (1-e^{-2 \theta \delta_{N}} )^{-1}$\\
 $+2   \alpha(E[\nu(\tau)]-\alpha) \frac{1}{\theta}  \delta_{N}  \displaystyle    e^{\theta {k_{N} \delta_{N}} -\theta i \Delta_{N}} (e^{\theta \delta_{N}}-1 )^{-1} $\\
 $\times [e^{-\theta k_{N} \delta_{N}} (e^{\theta k_{N} \delta_{N}}-1+k_{N}-k_{N} e^{\theta \delta_{N}} )+k_{N} (e^{\theta \delta_{N}}-1 )+  e^{\theta \delta_{N}} (e^{-\theta {k_N \delta_N}}-1 ) ]$\\
 $ +2 \frac{\gamma^{2}}{\theta}(E[\nu(\tau)]-\alpha) \frac{1}{\theta} \delta_{N}  e^{-\theta i \Delta_{N}}    (e^{\theta k_N \delta_{N}}-1+k_{N}-k_{N} e^{\theta \delta_{N}} ) (e^{\theta \delta_N}-1 )^{-1}$\\
 $   +\gamma^{2} \alpha  \frac{1}{\theta^{3}}  ( e^{-\theta k_N \delta_{N}}-1+k_{N}-k_{N} e^{-\theta \delta_{N}} ) +   \alpha^{2} \delta_{N}^{2} (k_{N}^{2}-k_{N} )$.

\vspace{0.25cm}

\textbf{II)} \textbf{Analytic expression of} $\boldsymbol{E[RV^2 (\tau+i\Delta_N-\Delta_N, k_N\delta_N)]$}

\vspace{0.25cm}
The analytic  expression of $E[RV^2 (\tau+i\Delta_N-\Delta_N, k_N\delta_N)] $ under Assumption \ref{Ass2} is easily obtained by replacing $i$ with $ i-1 $ in the explicit expression of $E[RV^2 (\tau+i\Delta_N, k_N\delta_N)] $ derived in {\bf I)}.

\vspace{0.5cm}
\textbf{IIIa)} \textbf{Analytic expression of} $\boldsymbol{E[RV (\tau+i\Delta_N, k_N\delta_N)RV (\tau+ i\Delta_N-\Delta_N\Delta_N, k_N\delta_N)]$} \textbf{for} $\boldsymbol{W_N \le \Delta_N}$

\vspace{0.25cm}
Assume that we are in the no-overlapping case $W_N  \le \Delta _N $. Then

\begin{eqnarray*}
&&E[RV (\tau+i\Delta_N, k_N\delta_N)RV (\tau+(i-1)\Delta_N, k_N\delta_N)] \\
&&=\displaystyle E\Big[ E\Big[\sum_{j=1}^{k_N}  \Big(\int_{a_{i,j-1,N}}^{a_{i,j,N}} \sqrt{v(s)} dW(s)\Big)^2 \displaystyle\sum_{j=1}^{k_N}  \Big(\int_{a_{i-1,j-1,N}}^{a_{i-1,j,N}} \sqrt{v(s)} dW(s)\Big)^2 | \mathcal{F}_{a_{i,k_N,N}}^v \Big]\Big] \\
&&=\displaystyle E\Big[ \sum_{j=1}^{k_N}  \int_{a_{i,j-1,N}}^{a_{i,j,N}} v(s) ds   \sum_{j=1}^{k_N}  \int_{a_{i-1,j-1,N}}^{a_{i-1,j,N}} v(s) ds \Big]=
 \displaystyle   \int_{a_{i,0,N}}^{a_{i,k_N,N}} \int_{a_{i-1,0,N}}^{a_{i-1,k_N,N}}E [v(s)E [ v(r)| \mathcal{F}_s^\nu] ]ds   dr, s<r. 
\end{eqnarray*}

\vspace{0.5cm}
 Under Assumption \ref{Ass2} (see, again, Appendix A in \cite{bz}),

\begin{eqnarray*}
 &&\displaystyle E[RV (\tau+i\Delta_N, k_N\delta_N)RV (\tau+(i-1)\Delta_N, k_N\delta_N)] \\
&& =\displaystyle\frac{1}{\theta^{2}} e^{\theta \Delta_{N}}\left(1-e^{\theta {k_{N}} \delta_{N}}\right)^{2}  e^{-2 \theta i \Delta_{N}}  \left[( E[\nu(\tau)]-\alpha)^{2}+\frac{\gamma^{2}}{\theta}\left(\frac{\alpha}{2}-E[\nu(\tau)]\right)\right] \\
&&\displaystyle- e^{-\theta \Delta_{N}}\left(2-e^{\theta {k_{N}} \delta_{N}}-e^{-\theta_{k_{N}} \delta_{N}}\right) \frac{\gamma^{2} \alpha}{2 \theta^3} \displaystyle-\frac{1}{\theta} k_{N} \delta_{N}\left(1-e^{\theta {k_N} \delta_{N}}\right) e^{-\theta {i} \Delta_{N}}\left[\left(\frac{\gamma^{2}}{\theta}+\alpha\right)(E[\nu(\tau)]-\alpha)\right] \\
&&\displaystyle-\frac{1}{\theta} k_{N} \delta_{N}\left(1-e^{\theta {k_{N}} \delta_{N}}\right) e^{\theta \Delta_N -\theta i\Delta_N}[\alpha(E[\nu(\tau)]-\alpha)]\displaystyle+\alpha^{2}\left(k_{N} \delta_{N}\right)^{2} .
\end{eqnarray*}

\vspace{1cm}

\textbf{IIIb)} \textbf{Analytic expression of} $\boldsymbol{E[RV (\tau+i\Delta_N, k_N\delta_N)RV (\tau+ i\Delta_N-\Delta_N\Delta_N, k_N\delta_N)]}$ \textbf{for} $\boldsymbol{W_N > \Delta_N}$

\vspace{0.25cm}
Assume now that we are in the overlapping case $W_N  > \Delta _N $. Then the parametric expression of  $E[RV (\tau+i\Delta_N, k_N\delta_N)RV (\tau+i\Delta_N-\Delta_N, k_N\delta_N)] $ can be decomposed $E[RV (\tau+i\Delta_N, k_N\delta_N)RV (\tau+(i-1)\Delta_N, k_N\delta_N)] $ into the sum of four components, that is

\vspace{0.5cm}

$E\Big[RV (\tau+i\Delta_N, k_N\delta_N)RV (\tau+(i-1)\Delta_N, k_N\delta_N)\Big] $

 $=E\Big[\Big(RV (\tau+i\Delta_N,\Delta_N)+RV (\tau+(i-1)\Delta_N, k_N\delta_N-\Delta_N)\Big)\Big(RV (\tau+(i-1)\Delta_N, k_N\delta_N-\Delta_N)+RV (\tau+i\Delta_N-k_N\delta_N, \Delta_N)\Big)\Big]$

$=E\Big[ RV (\tau+i\Delta_N,\Delta_N)RV (\tau+(i-1)\Delta_N, k_N\delta_N-\Delta_N)\Big] + E\Big[(RV (\tau+i\Delta_N,\Delta_N)RV (\tau+i\Delta_N-k_N\delta_N, \Delta_N)\Big]+ E\Big[RV^2 (\tau+(i-1)\Delta_N, k_N\delta_N-\Delta_N) \Big]$
$ +E\Big[RV (\tau+(i-1)\Delta_N, k_N\delta_N-\Delta_N)RV(\tau+i
\Delta_N-k_N\delta_N, \Delta_N) \Big]$.

\vspace{0.25cm}
We then obtain the parametric expressions of these four components, which we term $O_1$, $O_2$, $O_3$ and $O_4$, respectively (we omit the intermediate steps, as are they are analogous to those followed in {\bf I)} and {\bf IIIa)}):

\begin{itemize}

\item[-]$\displaystyle O_1:=E\Big[(RV (\tau+i\Delta_N,\Delta_N)RV (\tau+i\Delta_N-k_N\delta_N, \Delta_N)\Big]$\\
$=\alpha^{2} \Delta_{N}^{2}-\left(E[\nu(\tau)]-\alpha\right)\left(\frac{\gamma^{2}}{\theta}+\alpha\right) \Delta_{N} \frac{1}{\theta} e^{-\theta i \Delta_{N}}\left(1-e^{\theta \Delta_{N}}\right)\\-  \alpha\left(E[\nu(\tau)]-\alpha\right) \Delta_{N} \frac{1}{\theta} e^{-\theta {i} \Delta_{N}+\theta {k_N} \delta_{N}}\left(1-e^{\theta \Delta_{N}}\right) $
$\displaystyle -\frac{\gamma^{2} \alpha}{2 \theta^3}    e^{-\theta {k_{N}} \delta_{N}} (2-e^{-\theta \Delta_N} -e^{\theta \Delta_{N}} )\\+  \Big[\frac{\gamma^{2}}{\theta} \Big(\frac{\alpha}{2}-E[\nu(\tau)]\Big )+ (E[\nu(\tau)]-\alpha )^{2} \Big] \frac{1}{\theta^{2}} e^{-2 \theta {i} \Delta_{N}} (1-e^{\theta \Delta_{N}} )^{2} e^{\theta {k_{N}} \delta_{N}};$

\vspace{0.25cm}
\item[-] $\displaystyle O_2:=E\Big[ RV (\tau+i\Delta_N,\Delta_N)RV (\tau+(i-1)\Delta_N, k_N\delta_N-\Delta_N)\Big]$\\
$=\alpha^{2} \Delta_{N} (k_{N} \delta_{N}-\Delta_{N} )+ (E[\nu(\tau)]-\alpha ) \Big(\frac{\gamma^{2}}{\theta}+\alpha \Big) (k_{N} \delta_{N}-\Delta_{N} ) \frac{1}{\theta} e^{-\theta {i} \Delta_{N}}  (e^{\theta \Delta_{N}}-1 )$\\
$+\alpha (E[\nu(\tau)]-\alpha ) \Delta_{N} \frac{1}{\theta} e^{-\theta ({i}-1) \Delta_{N}} (e^{\theta ({k_{N}} \delta_{N}-\Delta_{N})}-1 ) $
$\displaystyle-\frac{\gamma^{2} \alpha}{2 \theta^{3}}   (1-e^{\theta \Delta_{N}} ) e^{-\theta {k_{N}} \delta_{N}}  (e^{\theta ({k_{N}} \delta_{N}- \Delta_N)}-1)$\\
$ -\displaystyle \Big[\frac{\gamma^{2}}{\theta} (\frac{\alpha}{2}-E[\nu(\tau)] )+ (E[\nu(\tau)]-\alpha )^{2} \Big] \frac{1}{\theta^{2}} e^{-2 \theta {i} \Delta_{N}} (1-e^{\theta \Delta_{N}} )e^{\theta \Delta_{N}} (e^{\theta ({k_{N}} \delta_{N}-\Delta_{N})}-1 );$

\vspace{0.25cm}
\item[-]$\displaystyle O_3:=E\Big[RV (\tau+(i-1)\Delta_N, k_N\delta_N-\Delta_N)RV(\tau+i
\Delta_N-k_N\delta_N, \Delta_N) \Big]$\\
$=\alpha^{2} \Delta_{N}\left(k_{N} \delta_{N}-\Delta_{N}\ )+ (E[\nu(\tau)]-\alpha\right)\ \Big(\frac{\gamma^{2}}{\theta}+\alpha \Big) \Delta_{N} \frac{1}{\theta} e^{-\theta (i-1) \Delta_{N}}(e^{\theta (k_{N} \delta_{N} - \Delta_{N})}-1 )$\\
$+\alpha (E[\nu(\tau)]-\alpha ) (k_{N} \delta_{N}-\Delta_{N} ) \frac{1}{\theta} e^{-\theta {i} \Delta_{N}+\theta{k_N} \delta_{N}}   (e^{\theta \Delta_{N}}-1 )+\frac{\gamma^{2} \alpha}{2 \theta^{3}} e^{\theta \Delta_{N}}(e^{\theta (k_{N} \delta_{N}-\Delta_N)}-1 ) e^{-\theta k_{N} \delta_{N}} (1-e^{-\theta \Delta_N} ) - \Big[\frac{\gamma^{2}}{\theta} \Big(\frac{\alpha}{2}-E[\nu(\tau)] \Big)\\+ (E[\nu(\tau)]-\alpha )^{2} \Big] \frac{1}{\theta^{2}} e^{-2 \theta i \Delta_{N}+\theta {k_{N}  \delta_{N}}} (1-e^{\theta \Delta_{N}} )e^{\theta \Delta_{N}} (e^{\theta ({k_{N} \delta_{N}} -  \Delta_{N})}-1 );$

\vspace{0.25cm}
\item[-] $\displaystyle O_4:=E\Big[RV^2 (\tau+(i-1)\Delta_N, k_N\delta_N-\Delta_N) \Big] \\=
 (k_{N} \delta_{N}-\Delta_{N}) \frac{1}{\delta_{N}} \Big[3 \frac{\gamma^{2}}{\theta^{2}} \alpha \Big(\frac{1}{2 \theta}(1-e^{-\theta \delta_{N}})\Big)^{2}+\frac{1}{\theta}-2\delta_{N} e^{-\theta \delta_{N}}     -\frac{1}{\theta} e^{-2 \theta \delta_{N}}+\delta_{N} (1+2 e^{-\theta \delta_{N}} )+\frac{1}{2 \theta} (e^{-2 \theta \delta_{N}}+4 e^{-\theta \delta_{N}}-5 ) + 3 \alpha^{2} \delta_{N}^{2} \Big] \\ 
 +e^{-2 \theta \delta_{N}} (1-e^{-2 \theta (k_{N} \delta_{N}-\Delta_{N} )} ) (1-e^{-2 \theta \delta_{N}})^{-1}   \Big[3 \frac{1}{\theta^2} (1-e^{-\theta\delta_N})^2 e^{-2 \theta {i} \Delta_{N}+2 \theta (1+k_{N} ) \delta_{N}} \Big( {E[\nu(\tau)]^{2}}-\frac{\gamma^{2}}{\theta} E[\nu(\tau)]-2 \alpha E[\nu(\tau)]+\frac{\gamma^{2} \alpha}{2 \theta}+\alpha^{2} \Big) \Big]\\
 +\frac{3}{\theta^{2}}e^{-\theta \delta_N} (1-e^{-\theta(k_{N} \delta_{N}-\Delta_{N} )} ) (1-e^{-\theta \delta_{N}})^{-1} e^{-\theta {i} \Delta_{N}+\theta (1+k_{N} ) \delta_{N}}\Big[ \frac{\gamma^{2}}{\theta}(E[\nu(\tau)]-\alpha )    (1-e^{-\theta \delta_{N}} )^{2}+ \gamma^{2}(E[\nu(\tau)]-\alpha )  \Big(\frac{1}{\theta}-2\delta_{N} e^{-\theta \delta_{N}}    -\frac{1}{\theta} e^{-2 \theta \delta_{N}} \Big)+2\alpha (E[\nu(\tau)]-\alpha )  (1-e^{-\theta \delta_{N}} ) \theta \delta_{N}\Big]\\
+\alpha^2\Big(k_N \delta_N-\Delta_N\Big)^2 -\alpha^2\delta_N \Big(k_N\delta_N - \Delta_N\Big)$\\
$\displaystyle +\frac{2}{\theta^2}e^{-2\theta i \Delta_N+2\theta k_N \delta_N-2 \theta \delta_{N}}(1-e^{\theta\delta_N})^2 (e^{\theta \delta_{N}}-1 )^{-1}(1-e^{-\theta \delta_{N}} )^{-1} (1-e^{-2 \theta \delta_{N}} )^{-1}\Big[\frac{\gamma^2}{\theta}(E[\nu(\tau)]-\alpha)^2 +  \Big(\frac{\alpha}{2}-E[\nu(\tau)]\Big)\Big]\\
\times \Big(1-e^{-\theta (k_N\delta_N  -\Delta_N)}\Big) \Big[-e^{\theta\delta_N}e^{-\theta(k_n\delta_N-\Delta_N)} +\Big(1+e^{-\theta(k_N\delta_N-\Delta_N)}\Big)-e^{-\theta \delta_{N} } \Big] $\\
\vspace{0.1cm}
 $\displaystyle -2 \frac{\gamma^2\alpha}{2 \theta^3}(2-e^{-\theta\delta_N }-e^{\theta \delta_N})  \Big(\Big(e^{-\theta (k_N\delta_N -\Delta_N)}-1\Big)+(k_N\delta_N -\Delta_N)\delta_N^{-1}(1-e^{-\theta \delta_N})\Big)(e^{-\theta \delta_{N}}+e^{\theta \delta_{N}}-2)^{-1}$ \\
 \vspace{0.1cm}
 $ -\frac{2}{\theta} \Big(\frac{\gamma^2}{\theta}+\alpha\Big) (E[\nu(\tau)]-\alpha)\delta_N (1-e^{\theta \delta_N})e^{-\theta i \Delta_N+  \theta  \Delta_N  }   (e^{\theta \delta_{N}}-1 )^{-2} \Big((e^{\theta (k_N\delta_N -\Delta_N) }-1)+ (k_N\delta_N -\Delta_N)\delta_N^{-1}(1- e^{\theta \delta_{N}})  \Big)$\\
 \vspace{0.1cm}
 $ -\frac{2\alpha}{\theta}  (E[\nu(\tau)]-\alpha)\delta_N (1-e^{\theta \delta_N}) e^{-\theta i \Delta_N+\theta k_N \delta_N} (e^{\theta \delta_{N}}-1 )^{-2} \displaystyle [(k_N\delta_N -\Delta_N)\delta_N^{-1} (e^{\theta \delta_{N}}-1)+e^{\theta \delta_{N}} (e^{-\theta {(k_N\delta_N -\Delta_N) }  }-1) ].$

\end{itemize}

The contribution to  the PSRV finite-sample bias due to the overlapping of consecutive local windows to estimate the spot volatility  (i.e., due to assuming that $W_N=k_N\delta_N>\Delta_N$) is mainly due to the terms $O_2$, $O_3$, and $O_4$. In fact, when $k_N\delta _N=\Delta_N$ (i.e., $W_N=\Delta_N$), the terms $O_2$, $O_3$, and $O_4$ are equal to zero.
Interestingly, the terms $O_2$, $O_3$, and $O_4$ are functions of the quantity $(k_N\delta_N-\Delta_N)$ (i.e., $W_N-\Delta_N$) and, in particular, 
are $O(k_N\delta_N-\Delta_N)$ as $(k_N\delta_N-\Delta_N)\to 0^+$ as $N \to \infty$, as one can check focusing on the terms $(k_N\delta_N-\Delta_N)$ and $(e^{\theta(k_N\delta_N-\Delta_N)}-1)$.  \\

\vspace{0.5cm}
After plugging the explicit expressions obtained in {\bf  I)},  {\bf II)} and  {\bf IIIa)} (resp., \textbf{IIIb)}) into Eq. (\ref{A1}),  simple but tedious calculations {yield} the parametric expression of $E[PSRV_{[\tau,\tau+h],N }]  $ under Assumption \ref{Ass2}, which can be expressed in the following compact form:

$$E[PSRV_{[\tau,\tau+h],N }]= \gamma^2 \alpha h A_N + \gamma^2  \Big( E[\nu(\tau)]-\alpha\Big) \frac{1-e^{-\theta h}}{h}B_N + C_N \ \  \ \ \ \ if \ \ \  W_N \le \Delta_N,  $$

$$E[PSRV_{[\tau,\tau+h],N }]= \gamma^2 \alpha h A_N + \gamma^2  \Big( E[\nu(\tau)]-\alpha\Big) \frac{1-e^{-\theta h}}{h}B_N + C_N + O_N \ \ \ \ \ \ if \ \ \ W_N > \Delta_N,$$

where:

\vspace{0.25cm}

 \begin{flalign}\label{AN}
 &\displaystyle A_N = (k_N \delta_N)^{-2} \Delta_{N}^{-1}  \Big\{\frac{2}{\theta} k_{N}  \Big[ \frac{3 }{2\theta^{2}} (1-e^{-\theta \delta_{N}} )^{2}+3 \frac{1}{\theta} \Big(\frac{1}{\theta}-2 e^{-\theta \delta_{N}} \delta_{N}-\frac{1}{\theta} e^{-2 \theta \delta_{N}} \Big)
\nonumber\\
& +3\frac{1}{\theta}  \delta_{N} (1+2 e^{-\theta \delta_{N}} ) + \frac{ 3}{2\theta^{2}} (e^{-2 \theta \delta_{N}}+4 e^{-\theta \delta_{N}}-5 ) \Big]+\frac{2}{\theta^{3}}(e^{-\theta {k_N} \delta_{N}}-1+k_{N}-k_{N} e^{-\theta \delta_{N}})\nonumber\\
&+\frac{1}{\theta^{3}} e^{-\theta \Delta_{N}}   (2-e^{\theta {k_{N} \delta_{N}}}-e^{-\theta {k_{N}} \delta_{N}} ) \Big\}    ;&&
\end{flalign}

 \begin{flalign}\label{BN}
&\displaystyle B_N =(k_N \delta_N)^{-2} e^{-\theta \Delta_{N}}(1-e^{-\theta \Delta_{N}})^{-1}  \Big\{  (1+e^{\theta\Delta_N}) \Big[ \frac{3}{\theta^{2}}(e^{\theta k_N\delta_N  }-1 )    (1-e^{-\theta \delta_{N}} )\nonumber\\
 &+ \frac{3}{\theta}( e^{\theta k_N \delta_{N}}-1 ) (1-e^{-\theta \delta_{N}} )^{-1} \Big(\frac{1}{\theta}-2 e^{-\theta \delta_{N}} \delta_{N}-\frac{1}{\theta} e^{-2 \theta \delta_{N}} \Big)\nonumber\\
 &+2 \frac{1}{\theta} \delta_{N} (e^{\theta \delta_{N}}-1 )^{-1} (k_{N}-1+e^{\theta k_N \delta_{N}}-k_{N} e^{\theta \delta_{N}} )\Big ]
+  \frac{2}{\theta} k_{N} \delta_{N} (1-e^{\theta k_N \delta_{N}}   ) \Big\}  ;&&
 \end{flalign}

  \begin{flalign}\label{CN}
&\displaystyle C_N =(k_N \delta_N)^{-2} \Bigg\{e^{-2 \theta \Delta_{N}}  (1-e^{-2 \theta h})(1-e^{-2 \theta \Delta_{N}})^{-1} \frac{1}{\theta^{2}} \Big[(E[\nu(\tau)]-\alpha)^{2}+\frac{\gamma^{2}}{\theta} \Big(\frac{\alpha}{2}-\nu(\tau) \Big) \Big] \nonumber\\
&\times\Big\{ (1+e^{2 \theta \Delta_{N}} ) (1-e^{-2 \theta \delta_{N}} )^{-1} \Big[3 (e^{2 \theta k_{N} \delta_{N}}-1 ) (1-e^{-\theta \delta_{N}} )^{2}+2 (1-e^{-\theta \delta_{N}} )\nonumber \\
&  +2 e^{\theta {k_{N}} \delta_{N}} (e^{-2 \theta    \delta_{N}}-1 )+2 e^{2 \theta {k_{N} \delta_{N}} -\theta \delta_{N}} (1-e^{-\theta {\delta_{N}}} ) \Big]  
 -2  e^{\theta \Delta_{N}} (1-e^{\theta {k_{N} \delta_{N}}})^{2}\Big\}  +(6 \alpha^{2} \delta_{N}^{2} k_{N} - 2 \alpha^{2} k_{N} \delta_{N}^{2}    ) h \Delta_{N}^{-1}\nonumber \\
 &+  e^{-\theta \Delta_N} (1-e^{-\theta h}) (1-e^{-\theta \Delta_N})^{-1}\Big\{\Big[6 \frac{\alpha}{\theta} \delta_{N} (E[\nu(\tau)]-\alpha) ( e^{\theta k_N \delta_N}-1) 
 +  2 \frac{\alpha}{\theta} \delta_{N}(E[\nu(\tau)]-\alpha) (e^{\theta \delta_{N}}-1 )^{-1}  \nonumber 
 \\
&\times [ (e^{\theta k_N \delta_{N}}-1+k_{N}-k_{N} e^{\theta \delta_{N}} )+ k_N  e^{\theta k_{N} \delta_{N}} (e^{\theta \delta_{N}}-1 )  +e^{\theta \delta_{N}} (1-e^{\theta {k_{N} \delta N}} ) ] \Big] (1+e^{\theta \Delta_N} )\nonumber\\
 & +  \frac{2\alpha}{\theta} k_{N} \delta_{N} (E[\nu(\tau)]-\alpha) (1+e^{\theta \Delta_{N}} ) (1-e^{\theta k_N \delta_{N}} ) \Big\}\Bigg\};&&
\end{flalign}

\begin{flalign}\label{ON}
&\displaystyle O_N= (k_N \delta_N)^{-2}\Bigg\{4\alpha ^2 h   \delta_N \nonumber\\
&\displaystyle + \gamma^2\alpha h \theta ^{-3}\frac{e^{-\theta(\delta_N +k_N\delta_N 
+\Delta_N)  }   }{{\Delta_N }}\Big[ e^{ \theta \delta_N }-2 e^{\theta\delta_N  (1+k_N)  }+e^{ \theta \delta_N  (1+2 k_N)}-2 e^{\theta(\delta_N +\Delta_N)  } \nonumber\\ & \displaystyle -4 e^{\theta(\delta_N  k_N+\Delta_N)  } k_N+e^{\theta(\delta_N +\delta_N 
k_N+\Delta_N)  } (2+k_N (4-6  \theta \delta_N ))\Big]\nonumber\\
&\displaystyle + 2 \gamma^2  \alpha h\theta ^{-3}\frac{ e^{-\theta\delta_N (1+k_N)  }}{{  \Delta_N }}\Big[ e^{\theta\delta_N  }+2 e^{\theta\delta_N k_N 
} k_N-e^{\theta\delta_N (1+k_N)  } (1-k_N (3\theta \delta_N-2  ))\Big]\nonumber\\
&-\displaystyle \gamma ^2\alpha h \theta^{-3} \frac{e^{-\theta (1+2 k_N)  \delta_N }  }{{\delta_N \Delta_N }}\Big[ -4 e^{2 \theta\delta_N k_N 
} (1-e^{ \theta \delta_N }) \Delta_N+6\theta e^{\theta \delta_N (1+2 k_N) } k_N\delta_N^2 \nonumber\\ & \displaystyle  +\delta_N \Big(e^{\theta(\delta_N+\delta_N k_N-\Delta_N)  }-2 e^{\theta(\delta_N+2 \delta_N k_N-\Delta_N)  }+e^{\theta(\delta_N+\delta_N k_N+\Delta_N) }+4 e^{2 \theta \delta_N k_N } k_N  -2 e^{\theta\delta_N(1+2 k_N)  } (2 k_N+3\theta \Delta_N  )\Big)\Big]\nonumber\\
&-\displaystyle \frac{1}{{ \theta ^2(1-e^{-2\theta \Delta_N  })}}\Big[e^{-2 \theta \Delta_N } (1-e^{-2  \theta h})
(-1+e^{\theta k_N\delta_N  }) (-2 e^{\theta\Delta_N  } (-1+e^{\theta k_N\delta_N  })\nonumber\\
&+\displaystyle\frac{(-3+e^{ \theta \delta_N }-e^{k_N  \delta_N \theta }+3 e^{\delta_N (1+k_N)
\theta }) (1+e^{2 \theta\Delta_N  })}{1+e^{  \theta \delta_N}} \Big(\frac{\gamma ^2 (\alpha -2 E[\nu(\tau)] )}{2 \theta }+(\alpha -E[\nu(\tau)] )^2\Big)\Big]\nonumber\\
&-\displaystyle 2\theta ^{-2}(\alpha -E[\nu(\tau)] )(1-e^{- \theta h} ) \frac{ (e^{\theta  k_N \delta_N}-1)  (\gamma ^2+\alpha  \theta(1+e^{\theta\Delta_N
 }k_N)  )\delta_N  }{{ (-1+e^{\theta\Delta_N  })}} &&\nonumber\\
 &+\displaystyle 2\theta^{-3} ( \alpha-E[\nu(\tau)] )  \frac{ e^{-\theta(\delta_N+k_N\delta_N -\Delta_N)  } (-e^{-\theta h}+1) }{ (-1+e^{\theta\delta_N  }) (-1+e^{\theta\Delta_N
 })  }\Big[\alpha  \theta ^2 e^{2\theta \delta_N (1+k_N)  } k_N \delta_N  -\alpha  \theta ^2 e^{\theta\delta_N (1+2
k_N)  } k_N \delta_N +\theta  (\gamma ^2+\alpha  \theta )e^{ \theta(\delta_N+ k_N\delta_N-\Delta_N) } k_N\delta_N \nonumber \\
& \displaystyle-\theta  (\gamma
^2+\alpha  \theta ) e^{\theta(\delta_N (2+k_N)-\Delta_N)  } k_N\delta_N  -e^{\theta\delta_N (2+k_N)  } \Big( \alpha  \theta ^2 (4+k_N)\delta_N+\gamma ^2 (6+\theta k_N\delta_N  -\theta\Delta_N  )\Big)\nonumber\\& \displaystyle+e^{\theta\delta_N (1+k_N)  } \Big(\alpha  \theta ^2
(4+k_N) \delta_N +\gamma ^2 (6+\theta\delta_N (4+k_N)  -\theta\Delta_N  )\Big) +e^{\theta(2 \delta_N (1+k_N)-\Delta_N)  } \Big( \alpha  \theta ^2 (4+k_N)\delta_N+\gamma ^2 (6+\Delta_N \theta)\Big)\nonumber\\
& \displaystyle-e^{\theta(\delta_N+2 k_N\delta_N -\Delta_N)  } \Big(\alpha  \theta ^2 (4+k_N)\delta_N +\gamma ^2 (6+4\theta \delta_N  +\theta\Delta_N  )\Big) \Big]\nonumber\\
&- \displaystyle \frac{2 \alpha ^2 \theta +\alpha  (\gamma ^2-4 \theta  E[\nu(\tau)] )+2 E[\nu(\tau)]  (\theta  E[\nu(\tau)]-\gamma ^2 )}{{2 \theta ^3 
(1+e^{\theta\delta_N  }) (-1+e^{2\theta \Delta_N  }) }}\Big[   (1-e^{-2  \theta h}) (-1+e^{ \theta k_N \delta_N}) (3-e^{\theta\delta_N  }+e^{\theta k_N \delta_N }-3 e^{ \theta (1+k_N)\delta_N }) (1+e^{2\theta
\Delta_N  })\Big]\nonumber\\
&-\displaystyle \frac{\theta ^{-3}(2 \alpha ^2 \theta +\alpha  (\gamma ^2-4 \theta  E[\nu(\tau)] )+2 E[\nu(\tau)]  (-\gamma
^2+\theta  E[\nu(\tau)] ) (1-e^{-2 h \theta }) }{{ (1+e^{\theta \delta_N }) (-1+e^{2\theta \Delta_N  }) }}\Bigg[-2 e^{2\theta k_N\delta_N 
 }+2 e^{ \theta (1+2 k_N)\delta_N }+e^{\theta\Delta_N  }+2 e^{2\theta  \Delta_N }+e^{\theta(\delta_N+\Delta_N)  }\\ &-2 e^{\theta(k_N\delta_N +\Delta_N)  }\nonumber  \displaystyle-2 e^{\theta(\delta_N+k_N\delta_N +\Delta_N)  }+e^{\theta(2  k_N\delta_N+\Delta_N) 
}+e^{\theta (\delta_N+2  k_N\delta_N+\Delta_N) }  -2 e^{\theta(\delta_N+2 \Delta_N)  }\Bigg] \Bigg\}.&&
 \end{flalign}

 The proof is complete.
\end{proof}

\subsection*{Theorem \ref{th4}}

\begin{proof}
Consider the exact parametric expression {for} the PSRV bias under Assumption \ref{Ass2} in the case $W_N > \Delta_N$, given in Lemma \ref{explbias}.  By expanding it sequentially, first as $\lambda \to 0$, and then as $h \to 0$, we obtain:

\vspace{0.5cm}
\hspace{-1cm} $E\Big[PSRV_{[\tau,\tau+h],N} - \langle\nu,\nu\rangle_{[\tau,\tau+h]}  \Big] =\begin{cases} \displaystyle  \Bigg(  \frac{4E[\nu(\tau)]^2}{\kappa^2\delta_N^{1+2b}} -\gamma^2 E[\nu(\tau)]    \Bigg)h  + O(h^{1-b}) +O(\lambda)  \quad \textit{if} \quad b\ge-1/2, c<-b \\ \\ -\gamma^2 E[\nu(\tau)]  h  + O(h^{-2b}) +O(\lambda)  \quad \textit{if} \quad b<-1/2, c<1+b \end{cases},$ 

as $\lambda \to 0$, $h \to 0$.

The sequential expansions as $h \to 0$, $\lambda \to 0$ are performed using the software \textit{Mathematica}. {The} code is available as supplementary material.

 \vspace{0.25cm}
Furthermore, let $(\mathcal{F}^{\nu}_t)_{t\ge0}$ denote the natural filtration associated {with} the process $\nu$.  It is straightforward to see that

\vspace{0.5cm}
$ E\Big[PSRV_{[\tau,\tau+h],N} - \langle\nu,\nu\rangle_{[\tau,\tau+h]}  | \mathcal{F}^{\nu}_\tau\Big] = \begin{cases} \displaystyle  \Bigg(  \frac{4 \nu(\tau)^2}{\kappa^2\delta_N^{1+2b}} -\gamma^2  \nu(\tau)     \Bigg)h  + O(h^{1-b}) +O(\lambda)  \quad \textit{if} \quad b\ge-1/2, c<-b \\ \\ -\gamma^2 \nu(\tau)   h  + O(h^{-2b}) +O(\lambda)  \quad \textit{if} \quad b<-1/2, c<1+b   \end{cases},$ 
 
  as $\lambda \to 0$, $h \to 0$.

\end{proof}

\subsection*{Theorem \ref{th1}}

\begin{proof}
Consider the exact parametric expression {for} the PSRV bias under Assumption \ref{Ass2} in the case $W_N \le \Delta_N$, given in Lemma \ref{explbias}. Then recall that for $N \to \infty$, $\Delta_N =O(\delta_N^c)$, $c \in (0,1),$ and $k_N=O(\delta_N^b) $, $b \in (-1,0)$. Moreover, note that for $b\ge-1/2$ and $c<-b$ or $b<-1/2$ and $c < 1+ b $, we have

$$\lim _{N \rightarrow+\infty} \frac{1}{k_{N} \Delta_{N}}=0$$ and 
$$\lim _{N \rightarrow+\infty} \frac{k_{N} \delta_N}{\Delta_{N}}=0.$$

Expanding $A_N$, $B_N$, and $C_N$ as $N\to \infty$, one obtains 

\begin{itemize}
\item[$\bullet$] $A_N \sim  1+ \frac{2}{\theta k_{N} \Delta_{N}}  + \frac{\theta(k_{N} \delta_{N})^2}{4\Delta_{N}}   -\frac{\theta^{2} (k_{N} \delta_{N} )^{2}}{4}-\frac{\theta\Delta_{N}}{2}  +\frac{\theta^{3}(k_{N} \delta_{N} )^{2}  \Delta_{N}}{8} ;   $
 
\item[$\bullet$] $B_N \sim 1+
\frac{4 }{\theta k_N \Delta_N }+
\frac{2 \delta_N}{\Delta_N}+
\frac{2}{k_N }-
\frac{4  \delta_N}{ k_N \Delta_N}-
\frac{k_N  \delta_N}{ \Delta_N }-
\frac{2\theta   \delta^2_N}{  \Delta_N  }+
\frac{1}{2}\theta \delta_N-
\frac{2\theta \delta_N}{k_N}+
\frac{1}{2}\theta \Delta_N+
  \frac{\theta\Delta_N}{k_N}+
  \frac{\theta \delta^2_N k_N}{2\Delta_N}-
  \frac{\theta^2\delta_N\Delta_N}{k_N}-\theta^2\delta_N^2+ \frac{\theta^2 \delta_N\Delta_N}{4}-\frac{\theta^3\delta_N^2\Delta_N}{2};$
\item[$\bullet$] $C_N \sim \frac{1-e^{-2 \theta h}}{2\theta^3}\Big[(E[\nu(\tau)]-\alpha)^2+\frac{\gamma^2}{\theta}(\frac{\alpha}{2}-E[\nu(\tau)])\Big] \Big[\frac{4 \theta^2  }{k_N \Delta_N} + \frac{4 \theta^3 \delta_N  }{\Delta_N} + \theta^4 \Delta_N+\frac{4 \theta^4 \Delta_N}{k_N} + 3 \theta^5 \Delta_N \delta_N\Big]+\frac{4 \alpha^2 h}{k_N \Delta_N}+\frac{8 \alpha (E[\nu(\tau)]-\alpha)(1-e^{-\theta h})}{\theta k_N \Delta_N},$
\end{itemize} from which  we get Eq.(\ref{dom}).

Based on the corresponding asymptotic expansions, one can easily check that as $N \to \infty$, if $b\ge-1/2$ and $c<-b$ or, alternatively,  $b<-1/2$ and $c < 1+ b $, then $A_N \to 1 $, $B_N \to 1$ and $C_N \to 0$.
This implies that as $N \to \infty$, if $b\ge-1/2$ and $c<-b$ or, alternatively,  $b<-1/2$ and $c < 1+ b $, then $E[PSRV_{[\tau,\tau+h],N }]= \gamma^2 \alpha h A_N + \gamma^2 (E [ \nu(\tau)]-\alpha) \frac{1-e^{-\theta h}}{h}B_N + C_N $ converges to $ E[\langle \nu,\nu\rangle_{[\tau,\tau+h]  }]=\gamma^2 \alpha h   + \gamma^2 (E[ \nu(\tau)]-\alpha) \frac{1-e^{-\theta h}}{\theta}   $, where the equivalence $E[\langle\nu,\nu\rangle_{[\tau,\tau+h]  }]=\gamma^2 \alpha h   + \gamma^2 (E [ \nu(\tau)]-\alpha) \frac{1-e^{-\theta h}}{\theta}$ is obtained from Appendix A in \cite{bz}.  

 \vspace{0.25cm}
  In particular, one can easily verify that, as $N \to \infty$:

 \begin{itemize}
 \item for $b \ge -1/2$ and $c<-b$,  
 \begin{eqnarray}\label{rateABCfirst}
 && A_N-1= O(\Delta_N),\,\, B_N-1=O(\Delta_N),\,\, C_N=O(\Delta_N) \quad \textit{if} \quad c <-b/2,  \\
&& A_N-1= O\Big(\frac{1}{k_N\Delta_N}\Big) ,\,\, B_N-1= O\Big(\frac{1}{k_N\Delta_N}\Big),\,\, C_N= O\Big(\frac{1}{k_N\Delta_N}\Big)\,\, \textit{if} \,  -b/2\le c  <-b;  
 \end{eqnarray}
 
 \item for $-2/3 \le b < -1/2$ and $c<1+b$,  
\begin{eqnarray}
 && A_N-1= O(\Delta_N), \,\, B_N-1 = O(\Delta_N),\,\, C_N = O(\Delta_N) \quad  \quad \textit{if} \quad c <(1+b)/2 ,\\
 && A_N-1= O(\Delta_N), \,\, B_N-1 =O\Big(\frac{k_N\delta_N}{ \Delta_N}\Big) ,\,\, C_N = O(\Delta_N)  \quad \textit{if} \quad (1+b)/2\le c <-b/2 , \\
 && A_N-1= O\Big(\frac{1}{k_N\Delta_N}\Big), \, B_N-1=O\Big(\frac{k_N\delta_N}{ \Delta_N}\Big), \, C_N=O\Big(\frac{1}{k_N\Delta_N}\Big)\, \textit{if} \,  -b/2\le c  <1+b  ; \nonumber\\
 \end{eqnarray}
 
   \item for $b<-2/3 $ and $c<1+b$,  
\begin{eqnarray}
 && A_N-1= O(\Delta_N), \,\, B_N-1= O(\Delta_N), \,\, \quad C_N = O(\Delta_N)   \quad \textit{if} \quad c<(1+b)/2,   \\
&& A_N-1= O(\Delta_N), \,\, B_N-1= O\Big(\frac{k_N\delta_N}{\Delta_N}\Big), \quad \,\, C_N = O(\Delta_N) \,\,  \textit{if} \,  (1+b)/2\le c  <1+b .\label{rateABClast}
 \end{eqnarray}

\end{itemize}

 The proof is complete.
  
\end{proof}

\subsection*{Corollary 1}
\begin{proof}
 Based on Eq. $(\ref{dom})$ and the asymptotic rates of $A_N$, $B_N$ and $C_N$ (see Eqs. $(\ref{rateABCfirst})-(\ref{rateABClast}))$, we observe that:

\begin{itemize}
\item[-] for $b \ge -1/2,$  $c <-b/2$ or $b<-1/2,$ $c < (1+b)/2$ or $b=-2/3,$ $c<1/6 $, 

$$E\Big[PSRV_{[\tau,\tau+h],N} - \langle\nu,\nu\rangle_{[\tau,\tau+h]}  \Big]=a_1\lambda\delta_N^c + o(\delta_N^c);$$

\item[-] for $b>-1/2, c \in (-b,-b/2),$

$$E\Big[PSRV_{[\tau,\tau+h],N} - \langle\nu,\nu\rangle_{[\tau,\tau+h]}  \Big]=a_2\frac{1}{\kappa \lambda   }\delta_N^{-b-c} + o(\delta_N^{-b-c} );$$

\item[-] for $b \in (-2/3,-1/2 ), c \in ((1+b)/2,1+b) \quad or \quad  {b <-2/3, c \in ((1+b)/2,1+b)},$

 $$E\Big[PSRV_{[\tau,\tau+h],N} - \langle\nu,\nu\rangle_{[\tau,\tau+h]}  \Big]=a_3\frac{\kappa }{\lambda }\delta_N^{1+b-c}  + o(\delta_N^{1+b-c}  );$$
 
\item[-] for $b=-2/3$, ${1/6<c<1/3} $,

 $$E\Big[PSRV_{[\tau,\tau+h],N} - \langle\nu,\nu\rangle_{[\tau,\tau+h]}  \Big]=a_3 \frac{\kappa}{ \lambda} \delta_N^{1/3-c}+ o(\delta_N^{1/3-c});$$

\item[-] for $b =-1/2,$ $c=1/4,$

 $$E\Big[PSRV_{[\tau,\tau+h],N} - \langle\nu,\nu\rangle_{[\tau,\tau+h]}  \Big]=\frac{1}{\lambda}\delta_N^{1/4}(a_1 \lambda^2 + a_2 \kappa^{-1} + a_3 \kappa)+ o(\delta_N^{1/4});$$
 
 \item[-] for $b=-1/2  , c >1/4,$

 $$E\Big[PSRV_{[\tau,\tau+h],N} - \langle\nu,\nu\rangle_{[\tau,\tau+h]}  \Big]=    \frac{1}{\lambda}\delta_N^{1/2-c}(a_2 \kappa^{-1} + a_3 \kappa) + o(\delta_N^{1/2-c}); $$ 
 
 \item[-] for $b=-2/3$, $c=1/6$,

 $$E\Big[PSRV_{[\tau,\tau+h],N} - \langle\nu,\nu\rangle_{[\tau,\tau+h]}  \Big]=  \delta_N^{1/6}( a_1 \lambda  + a_3 \kappa \lambda^{-1}) + o(\delta_N^{1/6}).$$

\end{itemize}

Thus, it is possible to select $\kappa$ and $\lambda$ such that the dominant term of the bias expansion is canceled only when $b=-1/2$ and $c\ge 1/4$ or $b=-2/3$ and $c=1/6$, provided that the selected values of $\kappa$ and $\lambda$ verify the condition $W_N \le\Delta_N$, which is equivalent to  $\kappa\delta_N^{1+b} \le \lambda \delta_N^c$.

The case $b=-1/2$ and $c=1/4$ is of particular interest, as it may allow to cancel the dominant term under the usual assumption $ \nu(0)=\alpha $, which is equivalent to $E[\nu(\tau)]=\alpha $. In fact, if $E[\nu(\tau)]=\alpha $, then $a_3=0$ and it is not possible to cancel the leading term of the bias expansion through the selection of $\kappa$ and $\lambda$ when $b=-1/2$ and $c>1/4$ or $b=-2/3$ and $c=1/6$. 

 Specifically, the leading term of the bias expansion in Eq. (\ref{dom}) can be canceled in the case  $b=-1/2$ and  $c=1/4$ if there exists a solution $(\tilde\kappa , \tilde\lambda )  \in \mathbb{R}_{>0} \times \mathbb{R}_{>0}$ to the following system

$$\begin{cases} \displaystyle a_3 \kappa^2 + a_1 \lambda^2 \kappa + a_2=0     \\  W_N \le \Delta_N  
 
 \end{cases},$$
 
 where $W_N=\kappa \delta_N^{1/2}$ and $\Delta_N = \lambda \delta_N^{1/4}$.  If a solution $(\tilde\kappa , \tilde\lambda ) \in \mathbb{R}_{>0} \times \mathbb{R}_{>0}$ exists, the corresponding bias-optimal selection of $W_N$ and $\Delta_N$ reads
 $$W_N=\tilde\kappa \delta_N^{1/2}, \quad\Delta_N=\tilde\lambda \delta_N^{1/4}.$$

\end{proof}

\subsection*{Lemma \ref{explbiasNoise}}
\begin{proof} Let Assumption \ref{Ass3} hold and consider the estimator:  

$ \widetilde{PSRV}_{[\tau,\tau+h],N } := \displaystyle \sum_{i=1}^{\lfloor h / \Delta_N \rfloor}  \Big[ \hat{\nu} (\tau+i\Delta_N)-\hat{\nu} (\tau+ i\Delta_N-\Delta_N ) \Big]^2, $

where, for $s$ taking values on the time grid of mesh-size $\delta_N$:

\begin{itemize}
\item[-]  $\hat{\nu} (s):=\widetilde{RV} (s, k_N\delta_N)(k_N\delta_N)^{-1},$  
\item[-]  $\widetilde{RV}(s, k_N\delta_N) := \displaystyle \sum_{j=1}^{k_N}  \Delta \tilde{p}^2 (s +  j\delta_N-k_N\delta_N, \delta_N ),$
\item[-] $\Delta \tilde p (s   )=\Delta   p (s){+}\Delta \eta (s   ):= \tilde p(s  )-\tilde p(s -\delta_N  )=    p(s  ) + \eta(s)-  p(s -\delta_N  )-  \eta(s -\delta_N  ).$ 

\end{itemize}

We observe that

$\widetilde{PSRV}_{[\tau,\tau+h],N }=\displaystyle  (k_N\delta_N)^{-2}\sum_{i=1}^{\lfloor h / \Delta_N \rfloor} [ \widetilde{RV} (\tau+i\Delta_N, k_N\delta_N) - \widetilde{RV} (\tau+(i-1)\Delta_N, k_N\delta_N)]^2$

$=\displaystyle  (k_N\delta_N)^{-2}\sum_{i=1}^{\lfloor h / \Delta_N \rfloor} \Big[ \sum_{j=1}^{k_N} \Delta \tilde{p}^2 (\tau+i\Delta_N + (j-k_N)\delta_N) - \sum_{j=1}^{k_N}\Delta \tilde{p}^2 (\tau+(i-1)\Delta_N + (j-k_N)\delta_N) \Big]^2 $

$=\displaystyle  (k_N\delta_N)^{-2}\sum_{i=1}^{\lfloor h / \Delta_N \rfloor} \Big[ \sum_{j=1}^{k_N} \Big(\Delta p (\tau+i\Delta_N + (j-k_N)\delta_N)+\Delta \eta (\tau+i\Delta_N + (j-k_N)\delta_N) \Big)^2$

$ \displaystyle- \sum_{j=1}^{k_N} \Big(\Delta p (\tau+(i-1)\Delta_N + (j-k_N)\delta_N)+\Delta \eta (\tau+(i-1)\Delta_N + (j-k_N)\delta_N) \Big)^2\Big]^2.$

\vspace{0.5cm}

 To simplify the notation, we replace $\Delta p (\tau+i\Delta_N +(j-k_N)\delta_N)$ with $r (i,j,N)$ and  
$\Delta \eta (\tau+i\Delta_N +(j-k_N)\delta_N)$
with $\epsilon(i,j,N)$ and rewrite:

\vspace{0.25cm}

$\widetilde{PSRV}_{[\tau,\tau+h],N }= \displaystyle  (k_N\delta_N)^{-2}\sum_{i=1}^{\lfloor h / \Delta_N \rfloor} \Bigg\{ \Bigg[ \sum_{j=1}^{k_N}\Big(r^2(i,j,N)- (r^2(i-1,j,N) \Big) \Bigg]^2$

$ \displaystyle+ \Bigg[ \sum_{j=1}^{k_N}\Big(\epsilon^2(i,j,N) - \epsilon^2(i-1,j,N) \Big)  \Bigg]^2+ 4\Bigg[ \sum_{j=1}^{k_N}\Big(r(i,j,N)\epsilon(i,j,N)- r(i-1,j,N)\epsilon(i-1,j,N)\Big)  \Bigg]^2$

$ \displaystyle + 2\sum_{j=1}^{k_N} \Big(r^2(i,j,N)- (r^2(i-1,j,N) \Big)\sum_{j=1}^{k_N}\Big(\epsilon^2(i,j,N) - \epsilon^2(i-1,j,N) \Big)$

$ \displaystyle +4\sum_{j=1}^{k_N}\Big(r^2(i,j,N) - r^2(i-1,j,N)\sum_{j=1}^{k_N}\Big(r(i,j,N)\epsilon(i,j,N)- r(i-1,j,N)\epsilon(i-1,j,N)\Big)$

$ \displaystyle +4\sum_{j=1}^{k_N}\Big(\epsilon^2(i,j,N) - \epsilon^2(i-1,j,N)\Big)\sum_{j=1}^{k_N}\Big(r(i,j,N)\epsilon(i,j,N)- r(i-1,j,N)\epsilon(i-1,j,N)\Big)\Bigg\}.$

Based on the previous expression, we can split the expected value of $ \widetilde{PSRV}_{[\tau,\tau+h],N }$ into the sum of the following six components:

\begin{itemize}
\item[i)]$   {PSRV}_{[\tau,\tau+h],N },$

\item[ii)]$\displaystyle(k_N\delta_N)^{-2}\sum_{i=1}^{\lfloor h / \Delta_N \rfloor} E\Bigg[ \sum_{j=1}^{k_N}\Big(\epsilon^2(i,j,N) - \epsilon^2(i-1,j,N) \Big)  \Bigg]^2,$

\item[iii)]$\displaystyle 4(k_N\delta_N)^{-2}\sum_{i=1}^{\lfloor h / \Delta_N \rfloor} E\Bigg[ \sum_{j=1}^{k_N}\Big(r(i,j,N)\epsilon(i,j,N)- r(i-1,j,N)\epsilon(i-1,j,N)\Big)  \Bigg]^2,$

\item[iv)]$\displaystyle 2(k_N\delta_N)^{-2}\sum_{i=1}^{\lfloor h / \Delta_N \rfloor} \Bigg\{ E\Big[ \sum_{j=1}^{k_N} \Big(r^2(i,j,N)- (r^2(i-1,j,N) \Big)\Big] E \Big[\sum_{j=1}^{k_N}\Big(\epsilon^2(i,j,N) - \epsilon^2(i-1,j,N) \Big)\Big]\Bigg\},$

\item[v)]$\displaystyle 4(k_N\delta_N)^{-2}\sum_{i=1}^{\lfloor h / \Delta_N \rfloor} E\Bigg[\sum_{j=1}^{k_N}\Big(r^2(i,j,N) - r^2(i-1,j,N)\Big)\sum_{j=1}^{k_N}\Big(r(i,j,N)\epsilon(i,j,N)- r(i-1,j,N)\epsilon(i-1,j,N)\Big)\Bigg],$

\item[vi)]$\displaystyle 4(k_N\delta_N)^{-2}\sum_{i=1}^{\lfloor h / \Delta_N \rfloor} E\Bigg[\sum_{j=1}^{k_N}\Big(\epsilon^2(i,j,N) - \epsilon^2(i-1,j,N)\Big)\sum_{j=1}^{k_N}\Big(r(i,j,N)\epsilon(i,j,N)- r(i-1,j,N)\epsilon(i-1,j,N)\Big)\Bigg].$
\end{itemize}

 Note that under Assumption \ref{Ass3}, $r$ is zero-mean  and $\epsilon$ is a zero-mean stationary process independent of $r$. Therefore components iv), v) and vi) are equal to zero. Moreover, note that the analytic expression of i) {was} already obtained in Theorem \ref{th1}. Thus, in order to obtain the analytic  expression of $ E[\widetilde{PSRV}_{[\tau,\tau+h],N }]$ under Assumption \ref{Ass3}, we only have to compute the analytic expressions of ii) and iii).

\vspace{0.25cm}
We start with ii). We have: 

$\displaystyle(k_N\delta_N)^{-2}\sum_{i=1}^{\lfloor h / \Delta_N \rfloor} E\Bigg[ \sum_{j=1}^{k_N}\Big(\epsilon^2(i,j,N) - \epsilon^2(i-1,j,N) \Big)  \Bigg]^2$

$=\displaystyle(k_N\delta_N)^{-2}\sum_{i=1}^{\lfloor h / \Delta_N \rfloor}  \Bigg\{\sum_{j=1}^{k_N}E\Bigg[ \epsilon^4(i,j,N) + \epsilon^4(i-1,j,N) -2\epsilon^2(i,j,N)\epsilon^2(i-1,j,N)   \Bigg] $

$ \displaystyle + 2\sum_{j=2}^{k_N} \sum_{h=1}^{j-1}E\Bigg[\Big(\epsilon^2(i,j,N)\epsilon^2(i,h,N) -\epsilon^2(i,j,N)\epsilon^2(i-1,h,N) - \epsilon^2(i-1,j,N)\epsilon^2(i,h,N)+  \displaystyle \epsilon^2(i-1,j,N)\epsilon^2(i-1,h,N)\Big) \Bigg]\Bigg\}$

$=\displaystyle(k_N\delta_N)^{-2}\sum_{i=1}^{\lfloor h / \Delta_N \rfloor}  \Bigg\{\sum_{j=1}^{k_N}2E[ \epsilon^4(i,j,N)] -2E[\epsilon^2(i,j,N)]^2\Bigg\}=  \frac{4\Big(Q_{\eta}+V^2_{\eta} \Big)h}{k_N\delta_N^{2}\Delta_N},  $ \\
 
since $\epsilon^2$ is an i.i.d. process such that $E[ \epsilon^2(i,j,N)]=2V_{\eta}$ and $E[\epsilon^4(i,j,N)]=2Q_{\eta}+6V_{\eta}^2$, as one can easily check.

\vspace{0.5cm}
Then we move {on} to iii). First, we rewrite:

$\displaystyle 4(k_N\delta_N)^{-2}\sum_{i=1}^{\lfloor h / \Delta_N \rfloor} E\Bigg[ \sum_{j=1}^{k_N}\Big(r(i,j,N)\epsilon(i,j,N)- r(i-1,j,N)\epsilon(i-1,j,N)\Big)  \Bigg]^2$

$=4\displaystyle(k_N\delta_N)^{-2}\sum_{i=1}^{\lfloor h / \Delta_N \rfloor}  \Bigg\{\sum_{j=1}^{k_N}E\Bigg[r^2(i,j,N)\epsilon^2(i,j,N) + r^2(i-1,j,N)\epsilon^2(i-1,j,N) - \displaystyle 2r(i,j,N)r(i-1,j,N)\epsilon(i,j,N)\epsilon(i-1,j,N)  \Bigg]$

$\displaystyle +  2\sum_{j=2}^{k_N} \sum_{h=1}^{j-1}E\Bigg[r(i,j,N)r(i,h,N) \epsilon(i,j,N) \epsilon(i,h,N) - r(i,j,N)r(i-1,h,N) \epsilon(i,j,N) \epsilon(i-1,h,N) $

$-r(i-1,j,N)r(i,h,N) \epsilon(i-1,j,N) \epsilon(i,h,N) +r(i-1,j,N)r(i-1,h,N) \epsilon(i-1,j,N) \epsilon(i-1,h,N) \Bigg]\Bigg\}.$

Then we note that:

\begin{itemize}

\item[-] $r$ is zero-mean and independent of $\epsilon$, therefore:
\begin{itemize}
\item[-] $E\Big[r(i,j,N)r(i-1,j,N)\epsilon(i,j,N)\epsilon(i-1,j,N)\Big]=0$;
\item[-] $E\Big[r(i,j,N)r(i,h,N) \epsilon(i,j,N) \epsilon(i,h,N) - r(i,j,N)r(i-1,h,N) \epsilon(i,j,N) \epsilon(i-1,h,N)$

$ -r(i-1,j,N)r(i,h,N) \epsilon(i-1,j,N) \epsilon(i,h,N) +r(i-1,j,N)r(i-1,h,N) \epsilon(i-1,j,N) \epsilon(i-1,h,N) \Big]=0;$

\end{itemize}

\item[-]$\epsilon$ is stationary and independent of $ r$, therefore:

$E\Big[r^2(i,j,N)\epsilon^2(i,j,N) + r^2(i-1,j,N)\epsilon^2(i-1,j,N)  \Big]=E\Big[\epsilon^2(i,j,N)\Big]E\Big[r^2(i,j,N) + r^2(i-1,j,N)  \Big]$

$=\displaystyle 2 V_{\eta} E\Big[\int_{a_{i,j-1,N}}^{a_{i,j,N}}v(s)ds + \int_{a_{i-1,j-1,N}}^{a_{i-1,j,N}}v(s)ds \Big],$ where $a_{i,j,N} := i\Delta_N + (j-k_N)\delta_N$.

\end{itemize}

Therefore, simple calculations allow rewriting component iii) as: 
 
$\displaystyle 4(k_N\delta_N)^{-2}\sum_{i=1}^{\lfloor h / \Delta_N \rfloor} E\Bigg[ \sum_{j=1}^{k_N}\Big(r(i,j,N)\epsilon(i,j,N)- r(i-1,j,N)\epsilon(i-1,j,N)\Big)  \Bigg]^2$

$= \displaystyle \frac{8V_{\eta}(\alpha-E[\nu(\tau)])(1+e^{-\theta\Delta_N})(1-e^{\theta k_N \delta_N})(1-e^{- \theta h})}{\theta(1-e^{-\theta\Delta_N}) k_N^2 \delta_N^2}+ \frac{16\alpha V_{\eta}h}{ k_N\delta_N\Delta_N}. $

\vspace{0.5cm}
Finally, putting everything together, we have 
$$\displaystyle E[\widetilde{PSRV}_{[\tau,\tau+h],N }] = E[ {PSRV}_{[\tau,\tau+h],N }]+ D_N,$$
where  
\begin{eqnarray*} 
&&D_N:=\displaystyle [4(Q_{\eta}+V^2_{\eta})+16\alpha V_{\eta}\delta_N]h\frac{1}{k_N \delta^2_N \Delta_N} +\frac{8}{ \theta} V_{\eta}(\alpha-E[\nu(\tau)])(1-e^{-\theta h})\frac{(1+e^{-\theta \Delta_N})(1-e^{-\theta k_N \delta_N})}{(1-e^{-\theta \Delta_N})k_N^2 \delta_N^2 }.\nonumber\\
\end{eqnarray*}
\end{proof}

{Analogous calculations in the overlapping case $W_N>\Delta_N$ lead to}

\begin{equation*}\begin{split}
D^*_N&=[   4\left(Q_\eta+V^2_\eta\right)+ 16\alpha V_\eta \delta_N] h\frac{1}{k_N^2\delta_N^3}+ \frac{8}{\theta}V_\eta(\alpha-E[\nu(\tau)]) (1-e^{-\theta h})\frac{1}{(1-e^{-\theta\Delta_N}) k_{N}^{ 2} \delta_{N}^{ 2}  } \Bigg\{ \frac{(2+  k_{N} )}{2k_N\delta_N}\Bigg[ \frac{(e^{\theta k_N \delta_N - \theta \Delta_N}-1)(k_N\delta_N+\Delta_N)}{k_N\delta_N-\Delta_N}+\\ &(e^{-\theta \Delta_N}-e^{\theta k_N\delta_N})\Bigg]+  \frac{k_N}{2\Delta_N}(1+e^{\theta k_N\delta_N})(1-e^{-\theta\Delta_N})\Bigg\}.\end{split}
\end{equation*}

\subsection*{Theorem \ref{th2}}
\begin{proof}

Consider $D_N$ in Lemma \ref{explbiasNoise}. Then recall that as $N \to \infty$, $\Delta_N =O(\delta_N^c)$, $c \in (0,1),$ and $k_N=O(\delta_N^b) $, $b \in (-1,0)$. Moreover, note that:

\begin{itemize}
\item[-]as $N \to \infty$, $E[ {PSRV}_{[\tau,\tau+h],N }] \to \langle\nu,\nu\rangle_{[\tau,\tau+h]}$  if  $b\ge-1/2$ and  $c<-b$ or $b<-1/2$ and $c < 1+ b $   (see Theorem \ref{th1});

\item[-]  as $N \to \infty$, $D_N \sim 4(Q_{\eta}+V^2_{\eta})h\frac{1}{k_N \delta_N^2 \Delta_N}+16\alpha V_{\eta}h\frac{1}{k_N \delta_N \Delta_N} +8 V_{\eta}(\alpha-E[\nu(\tau)])(1-e^{-\theta h})(1+e^{-\theta \Delta_N})\frac{1}{k_N \delta_N \Delta_N},  $

  thus  $D_N \to \infty$ as $N \to \infty$ for any $(b,c) \in (-1,0)\times (0,1)  $. In particular, as $N \to \infty$, $D_N$ is $O(\frac{1}{k_N \delta_N^2 \Delta_N})$ for any $(b,c) \in (-1,0)\times (0,1).  $

\end{itemize}

Therefore, as $N \to \infty$, if  $b\ge-1/2$ and  $c<-b$ or $b<-1/2$ and $c < 1+ b $, then $\displaystyle E[\widetilde{PSRV}_{[\tau,\tau+h],N }] = E[ {PSRV}_{[\tau,\tau+h],N }]+ D_N$ diverges, with rate $\frac{1}{k_N \delta_N^2 \Delta_N}$. The proof is complete.

\end{proof}

\subsection*{Theorem \ref{th3}}
\begin{proof} Recall from Definition \ref{larv} that for $\tau$ {with} values on the price-sampling grid of mesh size $\delta_N$: 

$\hat\nu(\tau) :=(k_N\delta_N)^{-1} \sum_{j=1}^{k_N} \Big[p(\tau-k_N\delta_N + j \delta_N)-p(\tau-k_N\delta_N + (j-1) \delta_N)\Big]^2.$

\vspace{0.25cm}
Moreover, from Appendix A in \cite{bz}, we have under Assumption \ref{Ass2}:

$\displaystyle    E\Big[\int_{\tau-\Delta }^{\tau} {\nu(t)} dt \Big]=\displaystyle   \alpha \Delta +  ( \nu(0)   -\alpha)\theta^{-1}e^{-\theta \tau} (e^{\theta \Delta} -1) $ and $ E [\nu(\tau)] = \alpha + ( \nu(0) -\alpha)e^{-\theta \tau}.$

\vspace{0.25cm}
Therefore, under Assumption \ref{Ass2},

$\displaystyle E[\hat\nu(\tau)-\nu(\tau)] =(k_N\delta_N)^{-1}  E\Big[\sum_{j=1}^{k_N} \Big[p(\tau-k_N\delta_N + j \delta_N)-p(\tau-k_N\delta_N + (j-1) \delta_N)\Big]^2\Big] - \Big[\alpha + ( \nu(0) -\alpha)e^{-\theta \tau}\Big]$

$=\displaystyle(k_N\delta_N)^{-1}  E\Big[\sum_{j=1}^{k_N} \Big[\int\limits_{\tau-k_N\delta_N + (j-1) \delta_N}^{\tau-k_N\delta_N + j \delta_N} \sqrt{\nu(t)} dW(t) \Big]^2\Big]- \Big[\alpha + ( \nu(0) -\alpha)e^{-\theta \tau}\Big]$

$=\displaystyle(k_N\delta_N)^{-1}   E\Big[\int_{\tau-k_N\delta_N }^{\tau} {\nu(t)} dt \Big] - \Big[\alpha + ( \nu(0) -\alpha)e^{-\theta \tau}\Big]$

$=\displaystyle ( \nu(0) -\alpha)e^{-\theta \tau} \Big[ (\theta k_N\delta_N)^{-1}    (e^{\theta k_N\delta_N} -1)-   1 \Big].$

Expanding {this} as $N \to \infty$, we can rewrite $E[\hat\nu(\tau)-\nu(\tau)] = \displaystyle (\nu(0)  -\alpha)e^{-\theta \tau}  \frac{1}{2}\theta k_N\delta_N   + o(k_N\delta_N).$ Furthermore, recall that $k_N\delta_N = O(\delta_N^{b+1})$ and $b \in (-1,0)$. Therefore,  under Assumption \ref{Ass2}, $E[\hat\nu(\tau)-\nu(\tau)] $ converges to zero  as $N \to \infty$, with rate $k_N\delta_N$.

\vspace{0.5cm}

Now let Assumption \ref{Ass3} hold and replace $p$ with $\tilde p$ in the definition of the locally averaged realized variance, i.e., consider the estimator $ w(\tau) :=(k_N\delta_N)^{-1} \sum_{j=1}^{k_N} \Big[\tilde p(\tau-k_N\delta_N + j \delta_N)-\tilde p(\tau-k_N\delta_N + (j-1) \delta_N)\Big]^2.$ Simple calculations lead to:

\vspace{0.5cm}
$ E[ w(\tau)-\nu(\tau)] =(k_N\delta_N)^{-1}  E\Big[\sum_{j=1}^{k_N} \Big[\tilde p(\tau-k_N\delta_N + j \delta_N)-\tilde p(\tau-k_N\delta_N + (j-1) \delta_N)\Big]^2\Big] - \Big[\alpha + ( \nu(0) -\alpha)e^{-\theta \tau}\Big]$

$=\displaystyle E[ \hat \nu(\tau)-\nu(\tau)]  + (k_N\delta_N)^{-1} \sum_{j=1}^{k_N} E\Big[ \Big[\eta(\tau-k_N\delta_N + j \delta_N)-\eta(\tau-k_N\delta_N + j \delta_N) \Big]^2 \Big] = \displaystyle E[ \hat \nu(\tau)-\nu(\tau)]  +  2 V_\eta\delta_N ^{-1}. $

Therefore, under Assumption \ref{Ass3}, $E[ w(\tau)-\nu(\tau)] $ diverges as $N \to \infty$, with rate $\displaystyle \frac{1}{\delta_N}$. The proof is complete.
 \end{proof}

\section{Indirect inference method for the feasible bias-optimal selection of local-window tuning parameter} 

 The feasible selection of $\kappa$ in equation (\ref{optimal_general}) requires, for a given $\beta$, the knowledge of the volatility process $\nu(t)$ at the instant $t=\tau$ and the vol-of-vol parameter $\gamma$.  A simple and computationally-efficient indirect inference method to obtain estimates of those quantities is as follows.

First, one estimates the spot volatility path using the fast Fourier transform algorithm, following the procedure detailed in Appendix B.5 of \cite{mrs}. In particular, from a given sample of log-price observations, one obtains estimates $\hat\nu_{n,N,M}$ of the spot volatility on the grid of mesh size $\Delta_M:=\frac{1}{2M+1}$, where $n$ denotes the sample size, while $N$ and $M$  denote, resp., the  cutting frequencies for the computation of Fourier coefficients of the volatility and the reconstruction of the spot volatility path.  See Chapter 4 in \cite{mrs} for the consistency of the estimator and guidance on the efficient selection of the cutting frequencies $N$ and $M$ for a given $n$.

 Then, using the reconstructed volatility path $\hat\nu_{n,N,M}(i)$, $i=1,2,\ldots,2M+1$, one infers the value of the parameter $\gamma$  by applying the following zero-intercept multivariate regression, based on the discretization of the CKLS process in Assumption \ref{Ass4}:

\begin{equation}\label{indinf}
\underline{Y}=\alpha\theta\,\Delta_M \underline{X}^1 - \theta \Delta_M\underline{X}^2+{\gamma}\sqrt{\Delta_M} \underline{Z},
\end{equation}
where $\underline{Z}$ is a vector of independent standard normal random variables, while the dependent variable $\underline{Y}$ and  independent variables $\underline{X}^1,$ $\underline{X}^2$ are defined as

$$Y_i:=\frac{\hat\nu_{n,N,M}(i+1)-\hat\nu_{n,N,M}(i)}{\hat\nu_{n,N,M}(i)^\beta},\quad X^1_i:=\hat\nu_{n,N,M}(i)^{-\beta}, \quad X^2_i:=\hat\nu_{n,N,M}(i)^{1-\beta}.$$

\noindent Denoting by $\hat\omega$  the estimate of the standard deviation  of the disturbance term, obtained from the regression residuals, we have $\hat\gamma= {\hat\omega}/\sqrt{\Delta_M}$.

 An estimate of $\nu(\tau)$ is simply given by the Fourier estimate of volatility  in correspondence of the beginning of the period of interest.

Finally, note that comparing the $R^2$ of the regression (\ref{indinf})  for different values of $\beta$ allows deciding which model under Assumption \ref{Ass4} fits the data better.

\end{appendices}

\end{document}